\documentclass[
reprint,
superscriptaddress,
nofootinbib,
amsmath,amssymb,
pra,
aps
]{revtex4-2}

\usepackage{dcolumn}
\usepackage{bm}
\usepackage[colorlinks,linkcolor=blue,urlcolor=blue,citecolor=blue]{hyperref}

\usepackage{physics}
\usepackage{mathtools}
\usepackage{dsfont}
\usepackage{amsthm}
\usepackage{amsmath}
\newcommand{\prob}{\mathbb{P}}

\newcommand{\dist}{\mathrm{dist}}
\newcommand{\diam}{\mathrm{diam}}
\newcommand{\ZZ}{\mathbb{Z}}

\newcommand{\sfN}{\mathsf{N}}
\newcommand{\sfE}{\mathsf{E}}
\newcommand{\mcN}{\mathcal{N}}
\newcommand{\mcE}{\mathcal{E}}

\newcommand{\mcC}{\mathcal{C}}
\DeclareMathOperator{\poly}{poly}
\DeclareMathOperator{\polylog}{polylog}

\newtheorem{theorem}{Theorem}
\newtheorem*{theorem*}{Theorem}

\newtheorem{lemma}[theorem]{Lemma}
\newtheorem{definition}[theorem]{Definition}
\newtheorem{proposition}[theorem]{Proposition}

\newtheorem{remark}[theorem]{Remark}

\begin{document}

\title{Proof of a finite threshold for the union-find decoder}

\author{Satoshi Yoshida}
\email{satoshiyoshida.phys@gmail.com}
\affiliation{Department of Physics, Graduate School of Science, The University of Tokyo, Hongo 7-3-1, Bunkyo-ku, Tokyo 113-0033 Japan}
\author{Ethan Lake}
\affiliation{
Department of Physics, University of California Berkeley, Berkeley, California 94720, USA
}
\author{Hayata Yamasaki}
\email{hayata.yamasaki@gmail.com}
\affiliation{
Department of Computer Science, Graduate School of Information Science and Technology, The University of Tokyo, 7-3-1 Hongo, Bunkyo-ku, Tokyo, 113-8656, Japan
}

\begin{abstract}
    Fast decoders that achieve strong error suppression are essential for fault-tolerant quantum computation (FTQC) from both practical and theoretical perspectives. The union-find (UF) decoder for the surface code is widely regarded as a promising candidate, offering almost-linear time complexity and favorable empirical error suppression supported by numerical evidence.     
    However, the lack of a rigorous threshold theorem has left open whether the UF decoder can achieve fault tolerance beyond the error models and parameter regimes tested in numerical simulations.
    Here, we provide a rigorous proof of a finite threshold for the UF decoder on the surface code under the circuit-level local stochastic error model. To this end, we develop a refined error-clustering framework that extends techniques previously used to analyze cellular-automaton and renormalization-group decoders, by showing that error clusters can be separated by substantially larger buffers, thereby enabling analytical control over the behavior of the UF decoder.
    Using this guarantee, we further prove a quasi-polylogarithmic upper bound on the average runtime of a parallel UF decoder in terms of the code size.
    We also show that this framework yields a finite threshold for the greedy decoder, a simpler decoder with lower complexity but weaker empirical error suppression.
    These results provide a solid theoretical foundation for the practical use of UF-based decoders in the development of fault-tolerant quantum computers, while offering a unified framework for studying fault tolerance across these practical decoders.
\end{abstract}

\maketitle

A decoder is one of the most essential components in fault-tolerant quantum computation (FTQC), responsible for identifying and correcting errors from syndrome measurement outcomes.
To be viable for real-time computation, where decoding must keep pace with syndrome generation to prevent the backlog problem~\cite{terhal2015quantum}, decoders must achieve both high error-correction capability and low computational complexity.
Beyond practicality, decoder efficiency also plays a foundational role in characterizing the overall spacetime overhead of FTQC~\cite{kovalev2013fault,gottesman2013fault,fawzi2018constant,yamasaki2024time,tamiya2025fault,10.1145/3717823.3718318,takada2025doubly}, which can be bottlenecked by the decoding step~\cite{terhal2015quantum,yamasaki2024time,tamiya2025fault,takada2025doubly}.
A fundamental requirement for any decoder used in FTQC is the existence of a finite threshold; if the physical error rate is below the threshold value, increasing the code size can reduce the logical error rate arbitrarily, as guaranteed by the threshold theorem~\cite{aharonov1997fault,knill1998resilient,10.5555/2011665.2011666}.
Representative types of decoders are supported by rigorous proofs of a finite threshold, including the minimum-weight perfect matching (MWPM) decoder~\cite{dennis2002topological,fowler2012proof} for surface codes, the small-set-flip decoder~\cite{fawzi2018constant} for quantum expander codes, and the renormalization group (RG) decoder~\cite{bravyi2013quantum} and cellular automaton (CA) decoders~\cite{harrington2004analysis,herold2015cellular,kubica2019cellular,vasmer2021cellular,balasubramanian2024local,lake2025fast,lake2025local} for topological codes.
In contrast, heuristic decoders, such as belief propagation combined with ordered statistics decoding (BP-OSD) for quantum low-density parity-check (LDPC) codes~\cite{panteleev2021degenerate,roffe2020decoding}, and machine learning decoders for topological codes~\cite{varsamopoulos2017decoding,PhysRevResearch.2.033399,bausch2024learning}, currently lack theoretical guarantees on error suppression, despite strong empirical performance.

The surface code is among the most prominent quantum error-correcting codes for FTQC due to its high threshold, locality, and the existence of efficient decoders~\cite{kitaev2003fault,dennis2002topological,fowler2012surface,google2023suppressing,bluvstein2024logical,google2025quantum,he2025experimental}.
The MWPM decoder achieves the highest known threshold among polynomial-time decoders for the surface code under the circuit-level local stochastic error model with a rigorous proof~\cite{dennis2002topological,fowler2012proof}, but it is prohibitively slow in many situations where decoding speed is important.
There are two approaches for faster decoders: one is to improve the runtime of the MWPM decoder using specific properties of the surface code~\cite{higgott2022pymatching,higgott2025sparse,wu2023fusion,wu2025micro,takada2025doubly}, and the other is to use approximate decoders with low computational complexity, such as the union-find (UF) decoder~\cite{delfosse2021almost} the greedy decoder~\cite{e17041946}, and the CA decoders~\cite{harrington2004analysis,herold2015cellular,kubica2019cellular,vasmer2021cellular,balasubramanian2024local,lake2025fast,lake2025local}.
The UF decoder stands out as a particularly promising candidate for practical implementation due to its almost-linear time complexity and favorable error-suppression performance, with efficient implementations on field-programmable gate arrays (FPGAs)~\cite{barber2025real,liyanage2024fpga,maurya2025fpga}.
The greedy decoder is also an efficient approximate decoder with linear time complexity, and its implementation in an FPGA is discussed in Refs.~\cite{holmes2020nisq+, ueno2021qecool, liao2023wit}.
For the surface code with code size $n$ and distance $d=O(\sqrt{n})$, the UF decoder can correct any error of weight up to $\lfloor (d-1)/2 \rfloor$~\cite{delfosse2021almost}, while the greedy decoder fails to correct certain errors of weight linear in $d$ due to the presence of Cantor-like non-correctable error chains~\cite{e17041946,hutter2015improved,reingold1981greedy}.
Despite the UF decoder's guarantee of correcting errors up to a given weight, such weight-based arguments are insufficient to establish the existence of a finite threshold, since at any fixed physical error rate, the expected number of errors per code block scales linearly with the block size.
While numerical simulations suggest threshold-like behavior for the UF decoder~\cite{delfosse2021almost}, no rigorous proof of a finite threshold exists under a general class of error models, such as the circuit-level local stochastic error model.

Existing proof techniques for establishing threshold theorems for the surface code are not readily applicable to the UF decoder.
Under the circuit-level error model, the decoder operates on syndromes arising from an $O\qty(d \times d \times d)$-sized window of the syndrome extraction circuit.
In the case of MWPM decoding, the proof of the threshold theorem~\cite{fowler2012proof} leverages the observation that a logical error can occur only if physical errors form a connected chain traversing opposite boundaries of the window.
Then, by crucially relying on the assumption that the MWPM decoder finds the optimal perfect matching whose weight is the global minimum in the window, one can argue that any such chain must contain at least $d/2$ errors to cause a logical error, thereby bounding the logical error rate by the probability of such long error chains occurring.
However, this optimality assumption is not guaranteed for the UF decoder, rendering such arguments inapplicable.
Alternative proof strategies, such as those used for CA decoders~\cite{harrington2004analysis,balasubramanian2024local,lake2025fast} and the RG decoder~\cite{bravyi2013quantum}, are based on an error-clustering argument showing that, with high probability, physical errors form clusters separated by error-free buffer regions, so that each error cluster can be corrected by these decoders.
However, the existing clustering bounds used in these works are insufficient for the UF decoder, because it is challenging to guarantee that the UF decoder identifies these error clusters correctly. Consequently, despite its practical appeal, establishing a rigorous threshold theorem for the UF decoder has remained an open problem in the field of FTQC\@.

In this work, we resolve this issue by rigorously proving the existence of a finite threshold for the UF decoder on the surface code under the circuit-level local stochastic error model.
To this end, we develop a refined error-clustering framework, based on a method originally introduced by G\'{a}cs~\cite{GACS198615,gacs2001reliable}.
In our approach,
error clusters are separated by substantially larger buffers than those considered in previous works on decoding~\cite{harrington2004analysis,bravyi2013quantum,balasubramanian2024local,lake2025fast}, which allows us to show that the UF decoder can identify and correct each highly isolated error cluster.
This refined framework allows us to prove a threshold theorem for the UF decoder, providing a rigorous guarantee of its fault tolerance that had previously been supported only by numerical evidence.
By using this framework, we further prove that a parallel implementation of the UF decoder has a quasi-polylogarithmic average runtime in the code distance $d$, i.e.,
$\exp\qty[O(\log\log d\cdot \log\log\log d)]$.
Previously, Ref.~\cite{10.5555/2685188.2685197} considered the problem of bounding the average runtime of the MWPM decoder for the surface code via discussing clusters of errors, but did not explicitly provide high-probability bounds on separations of error clusters needed for a complete analysis; in contrast, we leverage our rigorous error-clustering framework to obtain a provable upper bound on the UF decoder's average runtime.
We also demonstrate that this framework applies to the greedy decoder, leading to a corresponding threshold theorem.
These results offer a solid theoretical foundation for the practical use of UF-based decoders in the development of fault-tolerant quantum computers, while also providing a framework for establishing a unified understanding of their fault tolerance.

\begin{figure}
    \centering
    \includegraphics[width=3.4in]{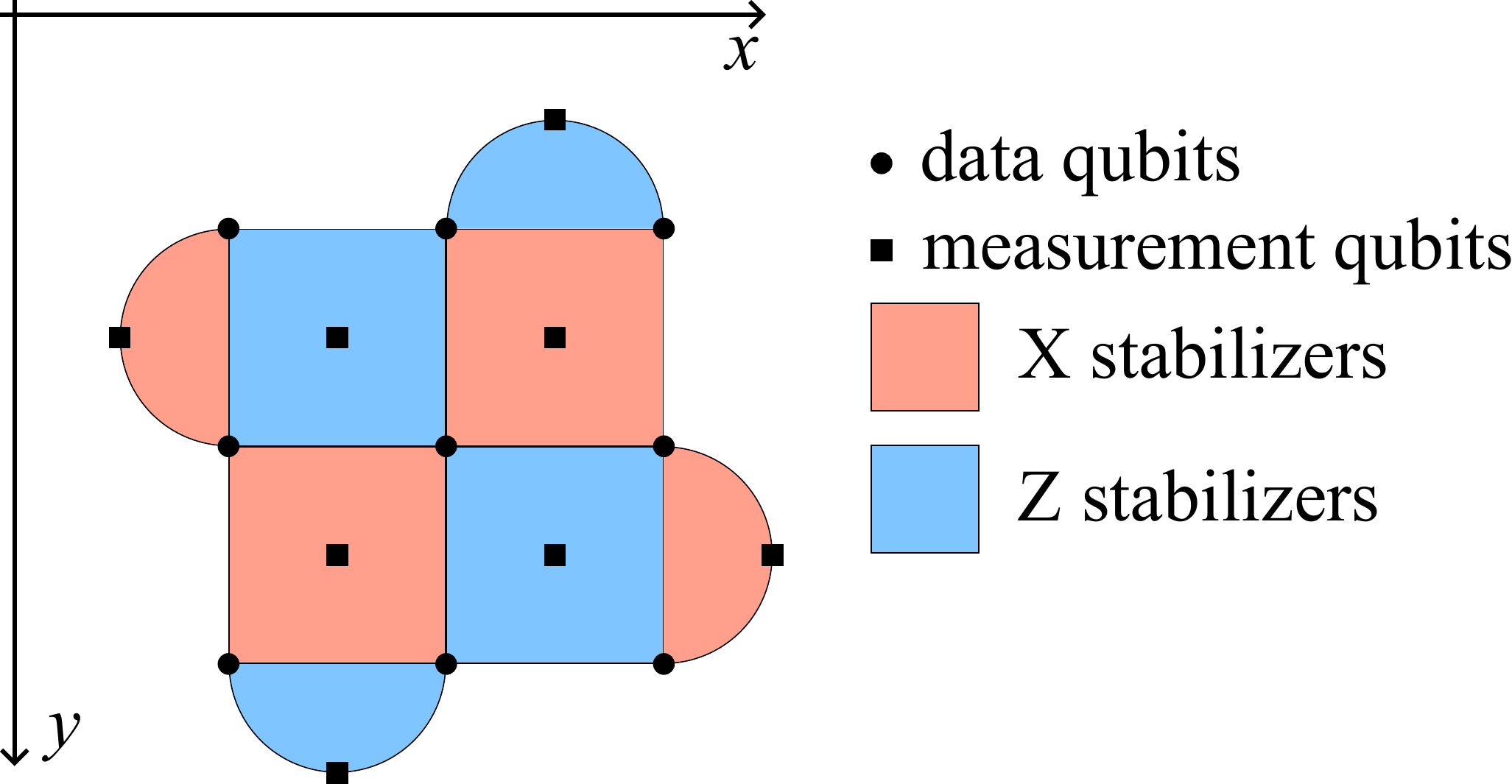}
    \caption{Illustration of a rotated surface code with distance $d=3$.
    The $x$- and $y$-axes are used in the definition of the stabilizer extraction circuit in Supplementary Information~\ref{appendix:circuit}.}
    \label{fig:surface_code_stabilizer}
\end{figure}

\section*{Results}
\paragraph*{Setting}
We consider a $[[n=d^2,1,d]]$ (rotated) two-dimensional surface code of distance $d$ defined on a $d\times d$ square lattice with open boundary conditions~\cite{PhysRevA.76.012305}, where each point on the lattice corresponds to a qubit, $n$ is the number of physical qubits, and it has $1$ logical qubit (see Fig.~\ref{fig:surface_code_stabilizer}).
It is a Calderbank-Shor-Steane (CSS) code defined by the stabilizer specified with generators in $\langle X_{f_X}, Z_{f_Z} \mid f_X\in F_X, f_Z\in F_Z\rangle$, where $F_X$ and $F_Z$ are disjoint subsets of faces of the lattice, and $X_{f_X}$ and $Z_{f_Z}$ are defined by $X_{f_X}\coloneqq \prod_{\vec{r}\in f_X} X_{\vec{r}}$ and $Z_{f_Z}\coloneqq \prod_{\vec{r}\in f_Z} Z_{\vec{r}}$, respectively.
We assign a measurement qubit for each face $f$, which is used for measuring the stabilizer generator $s_f$, where $s_f$ is defined by $s_f\coloneqq X_f$ if $f\in F_X$ and $s_f\coloneqq Z_f$ if $f\in F_Z$.
We let $\vec{r}_f$ denote the spatial coordinate of the measurement qubit corresponding to a face $f$.
To compute the logical error rate, we consider, for simplicity, a memory experiment, in which we prepare the logical $\ket{0}_L$ state, perform $d$ rounds of syndrome extraction circuits, and measure all data qubits in the $Z$ basis at the end of the experiment.
Our proof of a threshold also applies to online decoding during logical computation, which will be detailed in the discussion.
We assume that errors occur in the syndrome extraction circuit according to a circuit-level local stochastic Pauli error model; i.e., a set $F$ of faulty locations in the circuit is randomly chosen in a way that (i) Pauli errors occur in $F$, and (ii) there exists a physical error rate $p$ such that any set $S$ of locations in the circuit is included in $F$ with a probability upper-bounded by $\mathbb{P}[F\supset S]\leq p^{\abs{S}}$.
The syndrome extraction circuit we use is defined as follows (see also Supplementary Information~\ref{appendix:circuit} for a concrete instantiation):

\begin{definition}[Syndrome extraction circuit]
A quantum circuit with measurement outcomes $m_{f,i}$ for each face $f\in F$ and each time step $i\in [d]=\{1,2,\ldots,d\}$ is called the syndrome extraction circuit if
it implements a projective measurement of the stabilizer generator $s_f$ associated with face $f$ at each time step $i\in [d]$.
\end{definition}

Decoding can be performed by finding pairs of active detectors in the detector graph, as defined below.
\begin{definition}[Detector graph]
Given a syndrome extraction circuit, we define the set of detectors by
$D = \{\hat{m}_{f,i} \mid f\in F, i\in [d]\}$,
where $\hat{m}_{f,i}$ is defined by (i) $\hat{m}_{f,i} \coloneqq m_{f,0}$ if $i=0$, and (ii) $\hat{m}_{f,i} \coloneqq m_{f,i-1} + m_{f,i} \mod 2$ if $i>0$.
If the syndrome extraction circuit has the property that any $X$ or $Z$ error in a faulty location flips one or two detectors in $D$, we define the detector graph by an undirected graph $G = (V,E)$, where $V$ is the set of vertices consisting of the detectors $\hat{m}_{f,i} \in D$ and boundary vertices, defined below:
\begin{itemize}
    \item if there exists an $X$ or $Z$ error activating two detectors $\hat{m}_{f,i}$ and $\hat{m}_{f',i'}$, we add an edge between $\hat{m}_{f,i}$ and $\hat{m}_{f',i'}$;
    \item if there exists an $X$ or $Z$ error activating one detector $\hat{m}_{f,i}$, we add an edge between $\hat{m}_{f,i}$ and the corresponding boundary vertex.
\end{itemize}
We assign the spacetime coordinate $(\vec{r}_f, i)$ to each vertex $\hat{m}_{f,i}\in D$ in the detector graph.
\end{definition}
The logical error rate $p_\mathrm{L}$ is defined by the probability that the measurement outcome of the logical $Z$ operator is flipped due to errors in the syndrome extraction circuit, followed by executing a reliable decoder.
Any $X$ or $Z$ error at a faulty location in the syndrome extraction circuit flips one or two detectors; a $Y\propto XZ$ error is treated as simultaneous $X$ and $Z$ errors.
For the surface code, each detector can be flipped by errors occurring at only a constant number of locations in the syndrome extraction circuit.
Thus, we define the detector graphs for $X$ and $Z$ errors corresponding to the syndrome extraction circuit, where each edge may become faulty with a probability that is at most a constant factor times the physical error rate $p$.
The decoding process is equivalent to finding a pair of active detectors in the detector graph.

The UF decoder~\cite{delfosse2021almost} operates as follows.
We initialize the set $\mcC$ of UF clusters to be the set of singleton sets including each active detector $v_i\in V$, i.e.,
$\mcC\coloneqq \{\{v_i\}\}_i$,
and grow each set of UF clusters according to the following procedure.
Each UF cluster $c\in \mcC$ is defined to be valid if it includes an even number of active detectors or a boundary vertex; initially, each UF cluster contains a single (i.e., odd) active detector and is therefore invalid.
The decoder proceeds by iteratively growing all invalid UF clusters in $\mcC$ by a half-edge in the detector graph.
To formally state this procedure, we can represent a cluster as a subset of $E\cup V$; then, the growth of a UF cluster $c$ is to include all vertices incident to edges in $c$ and all edges incident to vertices in $c$.
If two (or more) UF clusters touch during the growth, they are merged.
This process continues until all UF clusters in $\mcC$ become valid.
Once all UF clusters become valid, an erasure decoder~\cite{delfosse2020linear} is applied to each valid UF cluster to find pairs of active detectors within the UF cluster.

\paragraph*{Main result}
We show that the UF decoder has a finite threshold under the circuit-level local stochastic error model, as stated in the following theorem.

\begin{theorem}[Threshold theorem for the UF decoder]
    \label{thm:threshold_theorem_UF}
    For the $[[d^2, 1, d]]$ rotated surface code, under the circuit-level local stochastic error model,
    there exist constants $p_\mathrm{th}>0$ and $\eta>0$ independent of $d$ such that, 
    if the physical error rate $p$ is below the threshold $p<p_\mathrm{th}$,
    then the logical error rate $p_\mathrm{L}$ of the UF decoder is upper-bounded by
    \begin{align}
        p_\mathrm{L}\leq \qty(p/p_\mathrm{th})^{\Omega\left(d^{\eta/\log\log d}\right)} \to 0~\text{as $d\to\infty$}.
    \end{align}
\end{theorem}

Our lower bound $d^{\eta/\log \log d}=\exp[\Omega\qty(\log d/\log \log d)]$ of the effective distance grows substantially faster than any polylogarithmic function $\polylog(d)=\exp[\Omega\qty(\log\log d)]$ as $d \to \infty$, while sub-polynomial in $d$ due to $\poly(d)=\exp[\Omega\qty(\log d)]$.
Although the UF decoder can correct any error of weight up to $\lfloor (d-1)/2 \rfloor$~\cite{delfosse2021almost}, establishing a threshold theorem purely from such weight-based correctability is challenging and, in particular, does not directly lead to bounds of the form $p_\mathrm{L}\leq (p/p_\mathrm{th})^{\Omega(d)}$.
For the surface code, the distance $d$ scales sublinearly with the code size $n$, so a constant physical error rate typically produces $O(n)$ errors within a code block, which is well beyond the regime captured by a weight-based argument alone.
By contrast, our key contribution is to develop an alternative theoretical framework that does not rely on error weight, and to use it to prove the existence of a finite threshold $p_\mathrm{th}>0$ such that the logical error rate vanishes $p_\mathrm{L}\to 0$ as $d\to \infty$, whenever the physical error rate $p$ satisfies $p<p_\mathrm{th}$.

To prove this result, we take inspiration from an approach originally introduced by G\'{a}cs~\cite{GACS198615,gacs2001reliable}, and extend the error-clustering framework previously used in Refs.~\cite{harrington2004analysis,bravyi2013quantum,balasubramanian2024local,lake2025fast}.
The core idea is to show that, with high probability, spacetime locality on the detector graph should ensure that the error configuration decomposes into well-separated clusters, with error-free buffer regions between them. These buffers should guarantee that the UF decoder always pairs active detectors within each cluster.
With this guarantee, we will show that if the diameter of such clusters of error configuration is smaller than the distance $d$, then no logical error can occur.
However, applying this proof strategy to the UF decoder is nontrivial, because the growth dynamics of UF clusters differ substantially from those of the CA and RG decoders, to which existing error-clustering frameworks apply~\cite{harrington2004analysis,bravyi2013quantum,lake2025fast}.
To address this, we develop a refined framework for the error-clustering analysis that ensures substantially sparser error configurations.

To formalize this, we define level-$k$ clusters of an error configuration as subsets of vertices with diameter at most $d_k+1$, which are separated from any level-$k'$ ($k'\geq k$) clusters of error configuration by a buffer of width at least $b_k-1$.
Note that the $\pm 1$ offsets ensure that $d_k+1$ and $b_k-1$ are valid upper and lower bounds irrespective of whether the relevant sets are defined using edges or vertices, absorbing the convention-dependent discrepancy.
The clusters of error configuration are not directly visible to the UF decoder, whose task is to estimate and correct them by forming UF clusters from active detectors. 
When a logical error occurs, there must exist a collection of error chains, connected by active detector pairings found by the UF decoder, whose length exceeds the distance $d$ (see Fig.~\ref{fig:error_chain}).
To capture this, we define extended UF clusters, which are formed by merging the UF clusters with the nearby level-$k$ clusters of error configuration; as in Fig.~\ref{fig:error_chain}, a logical error can occur only if there exists a level-$k$ extended UF cluster whose diameter exceeds the code distance $d$ for some $k$. 
The proof, therefore, reduces to showing that, with high probability, the diameter of every level-$k$ extended UF cluster remains below $d$.

\begin{figure}
    \includegraphics[width=3.4in]{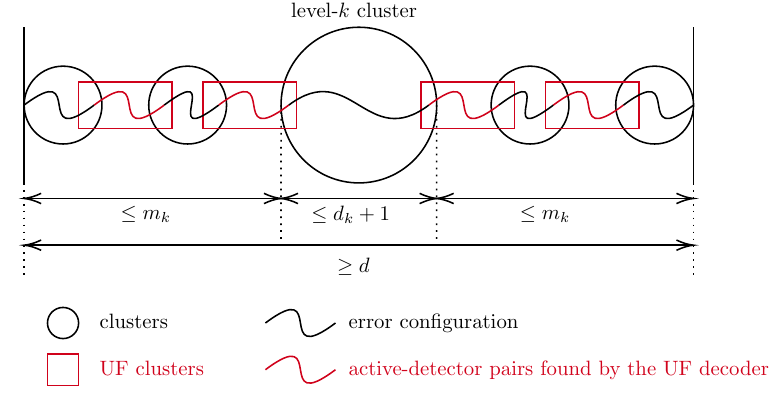}
    \caption{The logical error can happen only if the extended UF cluster has a diameter larger than or equal to $d$.}
    \label{fig:error_chain}
\end{figure}

Within each step in the UF decoder, the extended UF clusters can grow through two mechanisms: (1) the growth of the UF clusters themselves, and (ii) mergers with other extended UF clusters.
Because mergers can cause extended UF clusters to grow rapidly, the existing clustering bounds do not directly apply to analytically control the growth of extended UF clusters.
To address this point, we show that it is possible to choose the diameter $d_k$ and buffer $b_k$ of level-$k$ clusters of error configuration to scale superexponentially $\exp\qty[O\qty(k\log k)]$ in $k$, which is substantially larger than the exponential scaling $\exp\qty[O(k)]$
used in the previous error-clustering frameworks for analyzing other decoders.
In spite of this, we can still ensure that the level-$k$ clusters of error configuration with large $k$ occur only rarely, with a probability that decays doubly exponentially $O\qty((p/p_\mathrm{th})^{2^k})$ in $k$ for a finite threshold value $p_\mathrm{th}>0$.
With this superexponentially large buffer, we show that the growth of the extended level-$k$ UF cluster stops before it can merge with any other larger level-$k'$ ($k'\geq k$) extended UF cluster; under this guarantee, the parity of each UF cluster contained in the level-$k$ extended UF cluster remains unchanged when it merges with other extended UF clusters.
Thus, the growth of all UF clusters contained in the extended level-$k$ UF cluster stops before all the active detectors from the level-$k$ cluster of error configuration in the extended level-$k$ UF cluster are absorbed and merged into a single UF cluster; therefore, the growth of UF clusters in the level-$k$ extended UF cluster stops within at most $t\leq d_k+1$ steps.

With this stopping guarantee in place, we can upper-bound the diameter of level-$k$ extended UF clusters as follows.
Let $m_k$ denote the margin of the extended level-$k$ UF clusters, defined as the maximum distance from the boundary of the level-$k$ cluster of error configuration to the boundary of the level-$k$ extended UF cluster.
Then, the diameter of the extended level-$k$ UF clusters is at most $d_k + 1 + 2m_k$ (see Fig.~\ref{fig:error_chain}).
The growth of the margin $m_k$ of the level-$k$ extended UF cluster at each step can be controlled by the sparsity of the smaller level-$k''$ ($k''<k$) extended UF clusters, and in particular, we show $m_k= O(d_k)$.
As a result, the diameter of the extended level-$k$ UF clusters is bounded by $O(d_k)$,
Recalling that the probability of having a level-$k$ cluster of error configurations is suppressed doubly exponentially in $k$, while $d_k$ is superexponential in $k$ in our analysis, we obtain the desired upper bound of the logical error rate by enforcing that the diameter of extended UF clusters should remain below the code distance $d$.
See Methods for details.

\paragraph*{Quasi-polylogarithmic average runtime of parallel UF decoders}
The UF decoder can be parallelized by computing the growth of each UF cluster independently and merging them when they touch, and running the erasure decoder for each valid UF cluster in parallel, as proposed in Ref.~\cite{liyanage2024fpga}.
More precisely, we assign a processor to each active detector and compute the growth of the UF cluster initialized from the active detector using the assigned processor.
When two or more UF clusters touch during the growth, we merge them by communicating between the corresponding processors, and continue the growth of the merged UF cluster using one of the processors.
Since the level-$k$ extended UF cluster grows only up to $t\leq d_k+1$ steps and contains at most $O(d_k^3)$ vertices, the runtime for computing the growth of each level-$k$ extended UF cluster is upper-bounded by $O\qty(\poly(d_k))$, where $d_k=\exp\qty[O\qty(k\log k)]$.
Note that Ref.~\cite{delfosse2020linear} proposes data structures for the UF decoder that yield an almost-linear time complexity in the size of the decoding window, but practical parallel implementations, such as Ref.~\cite{liyanage2024fpga}, do not necessarily achieve this optimized scaling.
In our analysis, we do not require this almost-linear runtime bound; rather, we allow any (potentially suboptimal) implementation, as long as every extended UF cluster of diameter $O(d_k)$ can be processed in parallel with per-cluster running time bounded by $O\qty(\poly(d_k))$.
The probability of having a level-$k$ cluster of error configurations decays doubly exponentially $O\qty((p/p_\mathrm{th})^{2^k})$ in $k$.
Owing to the doubly exponentially decaying tail, the potentially large time required to process rare level-$k$ clusters for large $k$ has a negligible contribution to the expectation value, and the average runtime of the parallel UF decoder is dominated by clusters of levels up to $\overline{k}=O\qty(\log \log d)$.
Using this guarantee, we obtain the following theorem (see Methods for the proof).
\begin{theorem}[Quasi-polylogarithmic average runtime of parallel UF decoders]
    \label{thm:parallel_UF_runtime}
    For the $[[n=d^2, 1, d]]$ rotated surface code, under the circuit-level local stochastic error model with physical error rate $p<p_\mathrm{th}$ below the threshold $p_\mathrm{th}$, the average parallel runtime of the UF decoder is upper-bounded by
    \begin{align}
        \exp\qty[O(\log\log n\cdot\log\log\log n)],
    \end{align}
    where the average is taken over all possible error configurations in the $O(d^3)$-sized decoding window.
\end{theorem}
Our upper bound of the average parallel runtime of the UF decoder is quasi-polylogarithmic in the code size $n$, which is substantially faster than any polynomial function $\poly(n)=\exp[O(\log n)]$ in $n$ but slightly slower than any polylogarithmic function $\polylog(n)=\exp[O(\log\log n)]$ in $n$.
We emphasize that this upper bound of the average runtime provides a rigorous theoretical guarantee, rather than an estimate based on numerical experiments or heuristic parallelization proposals such as Ref.~\cite{liyanage2024fpga}.

\paragraph*{Extension to the greedy decoder}
The error-clustering argument is also applicable for the greedy decoder~\cite{holmes2020nisq+, ueno2021qecool, liao2023wit}, which pairs the closest active detectors in a greedy manner.
We can show a similar threshold theorem for the greedy decoder as follows.

\begin{theorem}[Threshold theorem for the greedy decoder]
    \label{thm:threshold_theorem_greedy}
    For the $[[d^2, 1, d]]$ rotated surface code, under the circuit-level local stochastic error model,
    there exists constants $p_\mathrm{th}>0$ and $\eta>0$ independent of $d$ such that, 
    if the physical error rate $p$ is below the threshold $p<p_\mathrm{th}$,
    then the logical error rate $p_\mathrm{L}$ of the greedy decoder is upper-bounded by
    \begin{align}
        p_\mathrm{L}=(p/p_\mathrm{th})^{\Omega(d^{\eta})} \to 0~\text{as $d\to\infty$}.
    \end{align}
\end{theorem}

The lower bound of the effective distance $\Omega(d^\eta)$ in this theorem is polynomial in $d$, which is better than the lower bound that we obtain for the UF decoder.
This difference arises for technical reasons: in the proof for the greedy decoder, we do not need to use the sparser clustering required for the UF decoder.
At the same time, the greedy decoder is believed to have worse empirical performance in error suppression than the UF decoder.
While the UF decoder can correct up to $\lfloor (d-1)/2 \rfloor$ errors~\cite{delfosse2021almost}, the greedy decoder causes a logical error in correcting Cantor-like error chains of weight $O(d^{\log_3 2})$~\cite{e17041946,hutter2015improved,reingold1981greedy}, implying that its effective distance is upper-bounded only by $O(d^{\log_3 2}) \approx O(d^{0.6309})$  (see also Methods for details of this example).
Nevertheless, our contribution is to prove threshold theorems for both the UF and greedy decoders through the same error-clustering framework, providing a unified understanding of their fault-tolerance properties.
An interesting direction for future work is to estimate and analyze tight asymptotic upper bounds on the logical error rate $p_\mathrm{L}$ of the UF decoder, which may require extremely large-scale numerical simulations; such numerics could clarify whether the UF decoder exhibits genuinely sublinear effective-distance scaling in the asymptotic large-$d$ regime (as the greedy decoder does), while establishing linear effective-distance scaling would likely require stronger analytical techniques.

\section*{Discussion}

We have proved the existence of a finite threshold for the UF decoder, as well as for the greedy decoder, on the surface code under the circuit-level local stochastic error model.
Using this framework, we have also proved that a parallel implementation of the UF decoder has a quasi-polylogarithmic average runtime in the size of the surface code.
For our proof, we developed a strengthened error-clustering framework that guarantees sparser clustering of error configurations than earlier techniques in the quantum error correction literature. 
As a result, we obtain a unified understanding of the fault-tolerance properties of these practical decoders within a single framework.
Despite the practical importance of the UF decoder, its threshold behavior had previously been supported only by numerical evidence.
By contrast, our results provide a rigorous theoretical guarantee of a finite threshold and a quasi-polylogarithmic upper bound on the average runtime, thereby justifying the use of UF-based decoding strategies in the development of fault-tolerant quantum computers.
Our results also suggest several further implications, as discussed below.

\paragraph*{Extension of the threshold theorem for the computational setting} 
This work shows the threshold theorem in the memory experiment setting with offline decoding, as in the existing proof of a finite threshold for the MWPM decoding~\cite{fowler2012proof}.
In this setting, the task of the decoder is to find pairs of active detectors in a $d\times d \times d$ spacetime region of the decoding window.
In a computational setting, however, fault-tolerant logical operations require online decoding.
To address this, we can employ the parallel window decoding~\cite{dennis2002topological,skoric2023parallel,PRXQuantum.4.040344,bombin2023modular}, where online decoding becomes feasible by introducing a sufficient buffer region in the decoding window.
With this modification, the decoding task reduces to finding the active-detector pairs within the $d\times d \times 3d$ spacetime region of the decoding window.
Importantly, this change does not affect the structure of the proof of the threshold theorem, nor the parameter choices that determine the threshold value.
With this parallel window decoding, our analysis directly applies to the protocol based on lattice surgery and gate teleportation~\cite{horsman2012surface}.

\paragraph*{Extension to hypergraph UF decoders}
In this work, for the distance-$d$ surface code, we assume that the syndrome extraction is repeated $d$ times, which is standard when decoding is performed independently for each code block.
By contrast, correlated decoding across multiple code blocks is proposed to reduce the number of syndrome extraction rounds per logical operation to $O(1)$, when logical CNOT gates are implemented transversally~\cite{bluvstein2024logical,cain2024correlated,zhou2025low,cain2025fast}.
Such correlated decoding naturally leads to decoding on a hypergraph rather than a detector graph, which may potentially change both the combinatorics of error configurations and the feasibility of decoding algorithms.
Whereas the MWPM decoder does not directly apply to the hypergraphs~\cite{cain2025fast}, it is proposed to apply the hypergraph UF decoder~\cite{cain2024correlated,delfosse2022toward} for this decoding task~\cite{cain2024correlated,zhou2025low}.
The growth of UF clusters in the UF decoder can be implemented on hypergraphs as well~\cite{cain2024correlated,delfosse2022toward}, and correspondingly, it is possible to argue extension of our error-clustering framework to the hypergraph setting.
However, our framework relies essentially on the spacetime locality of the underlying detector graph, which may fail for detector hypergraphs when long-range transversal CNOT gates are allowed across different code blocks.
Thus, a rigorous argument for the threshold theorem for hypergraph UF decoders within our framework requires an assumption on the spacetime locality of logical CNOT gates between code blocks.
We leave for future work whether our framework applies to the correlated decoding setting without such an additional assumption.

\paragraph*{Extension of other codes} 
Our proof techniques for the threshold theorem for the UF decoder for the surface code, in principle, apply to quantum low-density parity-check (QLDPC) codes satisfying the appropriate spacetime locality condition. For example, it would be promising to use our framework to analyze various decoders for higher-dimensional surface or color codes. QLDPC codes that do not satisfy spacetime locality, such as quantum expander codes~\cite{6671468,7354429}, lifted-product codes~\cite{panteleev2021degenerate,9567703}, bivariate bicycle codes~\cite{bravyi2024high}, and asymptotically good QLDPC codes~\cite{10.1145/3519935.3520017,9996782,10.1145/3564246.3585101}, do not directly satisfy our assumptions. Nevertheless, it would be interesting to identify alternative conditions under which the error-clustering framework can be applied to these codes in future work.

\section*{Methods}

\paragraph*{Error-clustering framework}
We use the error-clustering approach to show the threshold, which is based on an idea originally due to G\'{a}cs~\cite{GACS198615,gacs2001reliable}.
To this end, we use the clustering argument on the detector graph $G = (V, E)$, which is similar to Theorem~4.3 in Ref.~\cite{lake2025fast}.
To apply the clustering argument to the detector graph in the circuit-level error model, where each edge corresponds to an error, we define the clustering on the set of edges $E$ instead of the set of vertices $V$ used in Ref.~\cite{lake2025fast}.
As discussed in the main text, under the circuit-level local stochastic error model with physical error rate $p$, errors at different circuit locations may result in the activation of the same pair of vertices associated with a given edge.
Accordingly, we introduce the notation $\tilde{p}$ to denote an upper bound on the probability that a given edge is faulty, where $\tilde{p}$ is at most a constant multiple of the physical error rate $p$ of the circuit, i.e., there exists a constant $\xi\geq 1$ such that
\begin{align}
\label{eq:tilde_p}
    \tilde{p}\leq \xi p.
\end{align}

Given a subset $N\subset E$ of edges representing faulty locations, we introduce level-$k$ clusters as follows (see Definition~\ref{def:cluster} in the Supplementary Information for the precise definition).
For parameters $d,b>0$, a $(d,b)$-cluster in a subset $N\subset E$ of edges is a connected subset of $N$ in the graph with diameter at most $d$, surrounded by a width-$b$ buffer region that contains no edges of $N$.
For level $k\in\{0,1,\ldots\}$, we choose increasing sequences of parameters $\qty{d_k}_{k=0,1,\ldots}$ and $\qty{b_k}_{k=0,1,\ldots}$.
We then define the residual set $\sfN_k$ of unclustered edges inductively by $\sfN_0=N$, with $\sfN_{k+1}$ given by the set of edges in $\sfN_k$ that are not contained in any $(d_k,b_k)$-cluster in $\sfN_k$.
The set $\sfE_k=\sfN_k\setminus\sfN_{k+1}$ consists of the edges belonging to level-$k$ clusters, with all edges in level-$k'$ ($k'<k$) clusters already removed.
For a general choice of $d_k$ and $b_k$, the following lemma provides a high-probability bound on the occurrence of level-$k$ clusters.

\begin{lemma}[High probability bound on clustering of errors on edges of graphs]
\label{lem:sparsity}
    Let $G = (V, E)$ be a graph, and assume that a set $N$ of faulty locations is randomly specified as a subset of the edges of $G$ in such a way that there exists a rate $\tilde{p}$ such that any subset $S$ of edges satisfies $\mathbb{P}\qty[N\supset S]\leq \tilde{p}^{\abs{S}}$.
    Suppose that there exist constants $\Lambda, \Delta>0$ such that the locality condition
    \begin{align}
    \label{eq:upper_bound_on_neighboring_edges}
        \abs{B_e(r)} \leq \Lambda r^\Delta
    \end{align}
    holds for all $e\in E$ and $r\in \ZZ_{>0}$, where $B_e(r)$ is the set of edges within distance $r$ from $e$, and that 
    \begin{align}
        \label{eq:bkdk_constraint}
        b_k\geq d_k>0, \quad d_{k+1} \geq 4d_k + 3b_k
    \end{align}
    hold for all $k\geq 1$.
    Then, the probability
    \begin{align}
        p_k\coloneqq \max_{e\in E} \prob [e\in \sfN_k]
    \end{align}
    is upper bounded by
    \begin{align}
        p_k \leq \tilde{p}^{2^k} \prod_{j=0}^{k-1} \left[\Lambda \qty(b_{k-j-1}+\frac{d_{k-j-1}}{2})^\Delta\right]^{2^j}.
    \end{align}
\end{lemma}
\begin{proof}
This directly follows from Lemma~\ref{lem:min_size_set} in the Supplementary Information.
\end{proof}

Lemma~\ref{lem:sparsity} is essential for a threshold theorem: if $p$ in Eq.~\eqref{eq:tilde_p} is below a nonzero threshold $p_\mathrm{th}>0$, then $p_k$ decreases doubly exponentially with $k$.
The value $p_{\rm th}$ obtained in this way is a rigorous lower bound on the true threshold, although such analytical lower bounds are usually not tight.
Existing threshold analyses~\cite{harrington2004analysis,bravyi2013quantum,balasubramanian2024local,lake2025fast} take $b_k$ and $d_k$ to grow with the same scaling up to constant factors.
By contrast, we show here that the threshold theorem still holds even when $b_k$ grows faster than $d_k$; more precisely, for a constant $\lambda$, we can take
$b_k=\Theta(\lambda^{f(k+1)})$ and
$d_k=\Theta(\lambda^{f(k)})$ even for an increasing function $f(k)\gg k$,
so that the buffer width $b_k$, with exponent $f(k+1)$, can be substantially larger than the cluster diameter $d_k$, with exponent $f(k)$, with their relative separation controlled by the growth rate of $f$.
For example, in our analysis of the UF decoder, we take $f(k)\approx k\log k$, while the framework itself permits even faster-growing choices such as $f(k)=k^\alpha$ for $\alpha>1$.

\begin{lemma}[Threshold theorem on error clustering]
\label{lem:error_clustring}
    Suppose that $b_k$ and $d_k$ are given by
    \begin{align}
        \label{eq:bkdk_scaling}
        b_k = \beta \lambda^{f(k+1)}+1, \quad d_k = \gamma \lambda^{f(k)}-1,
    \end{align}
    using a scale parameter $\lambda>1$, positive constants $\beta, \gamma>0$, and a monotonically increasing function $f:\ZZ_{\geq 0}\to \mathbb{R}_{\geq 0}$ satisfying
    \begin{align}
        \label{eq:f_convergence}
        \sum_{j=0}^{k-1} f(k-j) 2^j = c\cdot 2^k + o(2^k),
    \end{align}
    where $c>0$ is a constant independent of $k$,
    and the constants are taken such that Eq.~\eqref{eq:bkdk_constraint} holds.
    Then, there exists a threshold $p_\mathrm{th}>0$ such that $p_k$ is upper bounded by
    \begin{align}
        p_k= O\qty(\qty(p/p_\mathrm{th})^{2^k}),
    \end{align}
    and the threshold $p_\mathrm{th}$ is given by
    \begin{align}
    \label{eq:analytical_threshold}
        p_\mathrm{th} = {1\over \xi \Lambda \lambda^{c\Delta} \left(\beta + {\gamma+\lambda^{-f(1)} \over 2}\right)^\Delta}.
    \end{align}
\end{lemma}
\begin{proof}
    Lemma~\ref{lem:sparsity}, together with Eq.~\eqref{eq:tilde_p}, shows that the probability $p_k$ is upper bounded by
    \begin{align}
        p_k &\leq \tilde{p}^{2^k} \prod_{j=0}^{k-1} \left[\Lambda \left(\beta \lambda^{f(k-j)} + {\gamma \lambda^{f(k-j-1)}+1 \over 2}\right)^\Delta\right]^{2^j}\\
        &\leq \qty(\xi p)^{2^k} \prod_{j=0}^{k-1} \left[\Lambda \left(\beta + {\gamma+\lambda^{-f(1)} \over 2}\right)^\Delta \lambda^{\Delta f(k-j)}\right]^{2^j} \label{eq:monotinicity_used}\\
        &= \qty(\xi p)^{2^k} \left[\Lambda \left(\beta + {\gamma+\lambda^{-f(1)} \over 2}\right)^\Delta\right]^{2^k-1} \lambda^{\Delta c\cdot 2^k + o(2^k)}\\
        &= O\qty((p/p_\mathrm{th})^{2^k}),
    \end{align}
    where we use the monotonicity of the function $f$ given by $f(k-j)\geq f(k-j-1)$ and $f(1)\leq f(k-j)$ for $j\leq k-1$ in Eq.~\eqref{eq:monotinicity_used}.
\end{proof}

\paragraph*{Spacetime locality of the detector graph}
For the rotated surface code, we can construct a syndrome extraction circuit whose detector graph satisfies the spacetime locality property---which means that the detector graph can be embedded in a finite-dimensional Euclidean spacetime---as shown by the following theorem. We present the details of the syndrome extraction circuit in Supplementary Information~\ref{appendix:circuit}.

\begin{theorem}[Spacetime locality of the detector graph]
    \label{thm:spacetime_locality}
The $[[d^2,1,d]]$ rotated surface code has a syndrome extraction circuit such that its detector graph $G = (V,E)$ satisfies the following spacetime locality property:
there exists a constant $C>0$ that does not depend on the code distance $d$ such that if two detectors $\hat{m}_{f,i}$ and $\hat{m}_{f',i'}$ are connected by an edge in $G$, then their spacetime coordinates $(\vec{r}_f, i)$ and $(\vec{r}_{f'},i')$ satisfy
\begin{align}
    \sqrt{\abs{\vec{r}_{f} - \vec{r}_{f'}}^2 + \abs{i-i'}^2} \leq C.
\end{align}
\end{theorem}

Due to this spacetime locality property, the detector graph satisfies the condition~\eqref{eq:upper_bound_on_neighboring_edges} in Lemma~\ref{lem:sparsity} with $\Lambda = O(1)$ and $\Delta = 3$.
We show an exact value of $C$ and $\Lambda$ in Supplementary Information~\ref{appendix:circuit}.
Note that our proof of a finite threshold for the UF decoder applies to any detector graph with the spacetime locality property~\eqref{eq:upper_bound_on_neighboring_edges} in Lemma~\ref{lem:sparsity}, and does not rely on other features specific to the surface code; in particular, the same proof applies to other topological codes as long as the spacetime locality property holds.

\paragraph*{Proof of the threshold theorem under the UF decoder}
We present the proof of the threshold theorem for the UF decoder using the clustering argument.
At any step $t$ of the growth of the UF clusters, we define the \emph{extended UF clusters} as follows.
As described in the main text, in each step, the UF decoder grows UF clusters by a half-edge in the detector graph~\cite{delfosse2020linear}.
We define the equivalence relation $\sim_t$ on the set of half-edges in the detector graph by the minimum equivalence relation satisfying the following conditions:
\begin{itemize}
    \item if two half-edges $e$ and $e'$ belong to the same UF cluster at step $t$, then $e\sim_t e'$ holds;
    \item if two half-edges $e$ and $e'$ belong to the same level-$k$ cluster of error configuration for some $k\geq 1$, then $e\sim_t e'$ holds.
\end{itemize}
Then, the extended UF clusters at step $t$ are defined by the equivalence classes of $\sim_t$.
We say that an extended UF cluster is a level-$k$ extended UF cluster if it includes a level-$k$ cluster and does not include any level-$k' (k'>k)$ clusters.
A level-$k$ extended UF cluster is called stable if all UF clusters in the level-$k$ extended UF cluster are valid and thus are not growing.
We let $\sfN_{k,t}$ denote the set of level-$k$ extended UF clusters at step $t$.
Note that $\sfN_{k,0} = \sfN_k$ holds at $t=0$.
We define a margin $m(\tilde{C}_t)$ of an extended UF cluster $\tilde{C}_t \in \sfN_{k,t}$ by the maximum distance between active detectors in the highest level cluster and the boundary of $\tilde{C}_t$.
We define the maximum margin $m_k$ of the level-$k$ extended UF clusters by
\begin{align}
    \label{eq:def_mk}
    m_k\coloneqq \max_t \max_{\tilde{C}_t\in \sfN_{k,t}} m\qty(\tilde{C}_t).
\end{align}
Then, the level-$k$ extended UF clusters satisfy the following lemma.

\begin{figure}
    \includegraphics[width=3.4in]{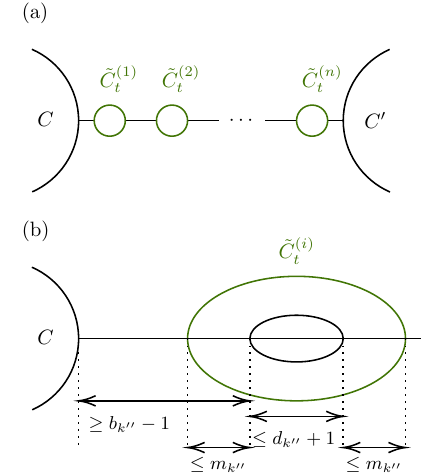}
    \caption{(a) Illustration of the distance chain in Eq.~\eqref{eq:distance_chain}.
    (b) The graphical explanation of the upper bound~\eqref{eq:fraction_upper_bound} on the fraction of the edges occupied by level-$k''$ UF clusters.}
    \label{fig:distance_chain}
\end{figure}

\begin{lemma}[Stopping guarantee on the growth of extended UF clusters]
    \label{lem:UF_clusters}
    We define the sequence $\qty{f_k}_{k=0,1,\ldots}$ by
    \begin{align}
        \label{eq:def_fk}
        f_k &\coloneqq
        \begin{cases}
            1 & (k=0)\\
            1-\sum_{k''=0}^{k-1} {(f_{k''}+1)(d_{k''}+1) \over (f_{k''}+1/2)(d_{k''}+1)+f_{k''} (b_{k''}-1)} & (k\geq 1)
        \end{cases}.
    \end{align}
    If $b_k$, $d_k$, and $f_k$ satisfy, for every $k$,
    \begin{align}
        \label{eq:constraint_bkdk_UF}
        d_k \geq 1, \quad f_k > {d_k+1\over b_k-1}>0,
    \end{align}
    then the following statements hold:
    \begin{enumerate}
        \item any level-$k$ extended UF cluster $\tilde{C}_t$ stops growing before merging with other level-$k' (k'\geq k)$ extended UF clusters;
        \item $m_k\leq {d_k+1\over 2f_k}$.
    \end{enumerate}
\end{lemma}
\begin{proof}
We show this lemma by induction on $k$.
For $k=0$, any level-$0$ cluster $C$ has diameter at most $d_0+1$, and buffer at least $b_0-1=f_0(b_0-1) > d_0+1$.
Therefore, any level-$0$ extended UF cluster $\tilde{C}_t$ stops growing before merging with other clusters, and $m_0\leq (d_0+1)/2 = (d_0+1)/2f_0$ holds.

For $k\geq 1$, we assume the statement holds for all $k'<k$.
By assumption, when a level-$k$ extended UF cluster $\tilde{C}_t$ overlapping with a level-$k$ cluster $C$ of error configuration merges with a level-$k''$ ($k''<k$) extended UF cluster $\tilde{C}''_t$, the level-$k''$ extended UF cluster $\tilde{C}''_t$ is stable.
Suppose that $\tilde{C}_t$ merges, at step $t$ before stopping the growth, with a level-$k'$ ($k'\geq k$) extended UF cluster $\tilde{C}'_t$ overlapping with a level-$k'$ cluster $C'$ of error configuration.
If $t>d_k+1$ holds, the level-$k$ extended UF cluster $\tilde{C}_t$ must be stable since all the level-$k$ clusters of error configuration in $\tilde{C}_t$ have diameter $d_k+1$ smaller than $t$ steps of growing UF clusters in $\tilde{C}'_t$.
Thus, for $\tilde{C}_t$ to keep growing, we necessarily have
\begin{align}
    \label{eq:bound_t}
    t\leq d_k+1;
\end{align}
in this case, there exists a set of level-$k_i$  ($k_i<k$) clusters $\left\{C^{(i)}\right\}_{i=1}^{n}$ such that (see also Fig.~\ref{fig:distance_chain}~(a))
\begin{align}
    &\dist\qty(C, \tilde{C}^{(1)}_t) + \sum_{i=1}^{n-1} \dist\qty(\tilde{C}^{(i)}_t, \tilde{C}^{(i+1)}_t) + \dist\qty(\tilde{C}^{(n)}_t, C') \nonumber\\
    &\leq t\leq d_k+1.
    \label{eq:distance_chain}
\end{align}

On the other hand, by the definition of the buffer width $b_k$, the distance between $C$ and $C'$ is given by
\begin{align}
    \label{eq:b_k_upper_bound}
    \dist\qty(C, C')\geq b_k-1.
\end{align}
Since the fraction of the edges occupied by level-$k_i$ UF clusters between $C$ of $\tilde{C}_t$ and $C'$ of $\tilde{C}_t'$ is at most [see also Fig.~\ref{fig:distance_chain}~(b)]
\begin{align}
    \label{eq:fraction_upper_bound}
    {d_{k''}+1+2m_{k''} \over d_{k''}+1+m_{k''}+b_{k''}-1},
\end{align}
the fraction of the edges not occupied by any level-$k''$ UF clusters for all $k''<k$ is at least
\begin{align}
    &1-\sum_{k''=0}^{k-1} {d_{k''}+1+2m_{k''} \over d_{k''}+1+m_{k''}+b_{k''}-1}\nonumber\\
    &\geq 1-\sum_{k''=0}^{k-1} {d_{k''}+1+{d_{k''}+1\over f_{k''}} \over d_{k''}+1+{d_{k''}+1\over 2f_{k''}}+b_{k''}-1}\\
    &= f_k,
    \label{eq:f_k_upper_bound}
\end{align}
where the last line follows from Eq.~\eqref{eq:def_fk}.
Due to the bounds~\eqref{eq:b_k_upper_bound} and~\eqref{eq:f_k_upper_bound}, the left-hand side of the inequality~\eqref{eq:distance_chain} satisfies
\begin{align}
    &\dist\qty(C, \tilde{C}^{(1)}_t) + \sum_{i=1}^{n-1} \dist\qty(\tilde{C}^{(i)}_t, \tilde{C}^{(i+1)}_t) + \dist\qty(\tilde{C}^{(n)}_t, C') \nonumber\\
    &\geq (b_k-1)f_k.
\end{align}

Therefore, as long as $(b_k-1) f_k > d_k+1$ holds, the inequality~\eqref{eq:distance_chain} does not hold, which implies that the level-$k$ extended UF cluster $\tilde{C}_t$ becomes stable before merging with any level-$k'$  ($k'\geq k$) extended UF cluster.

Similarly to Eq.~\eqref{eq:f_k_upper_bound}, at each step of the growth of the level-$k$ extended UF cluster $\tilde{C}_t$, for any path strating from an active detector in the level-$k$ cluster of error configuration in $\tilde{C}_t$ to the boundary of $\tilde{C}_t$, the fraction of the edges not occupied by any level-$k''$ ($k''<k$) UF clusters is at least $f_k$.
Therefore, the margin $m\qty(\tilde{C}_t)$ of the level-$k$ extended UF cluster $\tilde{C}_t$ grows at most by a factor of $1/f_k(\geq 1)$ at each step of the growth of UF clusters in $\tilde{C}_t$, i.e.,
\begin{align}
    m\qty(\tilde{C}_t) \leq {t\over 2f_k}.
\end{align}
Using Eq.~\eqref{eq:bound_t}, we obtain from the definition~\eqref{eq:def_mk} of $m_k$
\begin{align}
\label{eq:m_k}
    m_k = \max_{t\leq d_k+1} \max_{\tilde{C}_t\in \sfN_{k,t}} m\qty(\tilde{C}_t) \leq {d_k+1 \over 2f_k},
\end{align}
which completes the proof.
\end{proof}

To exploit this stopping guarantee on the UF decoder together with the error clustering result in Lemma~\ref{lem:error_clustring},
we take the sequences $\{(b_k, d_k)\}_{k=0,1,\ldots}$ as
\begin{align}
\label{eq:d_k_uf}
    b_k = \beta \lambda^{(k+2) \log (k+2)} +1, \quad d_k = \gamma \lambda^{(k+1) \log (k+1)}-1,
\end{align}
which corresponds to
$f(k) = (k+1) \log (k+1)$ in Eq.~\eqref{eq:bkdk_scaling}, satisfying Eq.~\eqref{eq:f_convergence} for $c=\sum_{n=2}^{\infty} {2 n\log n \over 2^{n}} \approx 3.57256$, where $\log$ is the logarithm to the base of Napier's constant $e \coloneqq \lim_{n\to\infty} (1+{1\over n})^n \approx 2.71828$.
We assume that the constants satisfy the condition that
\begin{align}
    \label{eq:constraints}
    \lambda>e, \quad {6\gamma \over \beta} [\zeta (\log \lambda)-1+\lambda^{-2\log 2}]\leq 1,
\end{align}
where $\zeta(s) \coloneqq \sum_{k=1}^{\infty} k^{-s}$ is the Riemann zeta function, which converges to a finite positive number for $s>1$.

Then, $f_k$ defined in Eq.~\eqref{eq:def_fk} satisfies $f_k \geq {1\over 2}$, which is shown by induction on $k$ as follows.
For $k=0$, $f_0 = 1 \geq {1\over 2}$ holds.
For $k\geq 1$, if $f_{k'}\geq {1\over 2}$ holds for all $k'<k$, we have
\begin{align}
    f_k
    &\geq 1-\sum_{k''=0}^{k-1} {{3\over 2} (d_{k''}+1) \over d_{k''}+1+{1\over 2} (b_{k''}-1)}\\
    &\geq 1-{3\gamma \over \beta}\sum_{k''=0}^{k-1} \lambda^{(k''+1)\log (k''+1) - (k''+2)\log (k''+2)}\\
    &\geq 1-{3\gamma \over \beta} \left(\lambda^{-2\log 2} + \sum_{k''=1}^{\infty}\lambda^{-\log (k''+1)}\right)\\
    &= 1-{3\gamma\over \beta} [\zeta (\log \lambda)-1+\lambda^{-2\log 2}]\\
    &\geq {1\over 2}.
\end{align}
Note that this bound on $f_k$ implies that $m_k$ in Eq.~\eqref{eq:m_k} satisfies $m_k = O(d_k)$.
The constraints in Eqs.~\eqref{eq:bkdk_constraint} and \eqref{eq:constraint_bkdk_UF} are satisfied when
\begin{align}
    \label{eq:constraints_2}
    2\lambda^{-2\log 2} \leq {\beta\over \gamma} \leq {1-4\lambda^{-2\log 2} \over 3}, \quad \gamma \geq 2.
\end{align}
Under the above choice of $\{(b_k, d_k)\}_{k=0, 1, \ldots}$ and the constants, we show the threshold theorem.

\begin{proof}[Proof of Theorem~\ref{thm:threshold_theorem_UF}]
By definition of the margin $m_k$ and the diameter $d_k$, when the growth of the clusters stops, the largest diameter of the level-$k$ extended UF clusters is at most $d_k+2m_k+1$ (see also Fig.~\ref{fig:distance_chain}~(b)).
A logical error can happen only if there exists a nontrivial loop in an extended UF cluster (see also Fig.~\ref{fig:error_chain}).
If
\begin{align}
    k\leq k_0&\coloneqq \max\{k \mid d_k+2m_k+1<d\}
\end{align}
holds, the diameter of any extended UF cluster cannot be larger than $d$, and no logical error occurs.
Since Lemma~\ref{lem:UF_clusters}, together with Eq.~\eqref{eq:d_k_uf}, yields
\begin{align}
    d_k+2m_k+1\leq d_k+2\frac{d_k+1}{2f_k}+1=O\qty(\lambda^{k\log k}),
\end{align}
the value of $k_0$ is evaluated as
\begin{align}
    k_0 &= \exp [W(\log_\lambda d)] + O(1),
\end{align}
where $W(x)$ is the Lambert W function defined by $W(x) e^{W(x)} = x$ such that $W(x) > 0$ for $x>0$.
Since $W(x)$ satisfies $W(x) = \log x - \log \log x + o(1)$ as $x\to \infty$, the logical error rate $p_\mathrm{L}$ is upper bounded by
\begin{align}
    p_\mathrm{L}&\leq \prob[\sfN_{k_0+1}\neq\varnothing]\\
    &\leq \sum_{e\in E} \prob[e\in \sfN_{k_0+1}]\\
    &\leq \abs{E}\cdot p_{k_0+1}\\
    &= O\left(d^3\right) \cdot (p/p_\mathrm{th})^{2^{\exp[W(\log_\lambda d)] + O(1)}}\\
    & = (p/p_\mathrm{th})^{2^{[1+o(1)]\log_\lambda d /\log \log_\lambda d + O(1)}}\\
    &= (p/p_\mathrm{th})^{\Omega\left(d^{\eta / \log \log d}\right)},
\end{align}
where $\eta$ is a positive constant satisfying $\eta < \log_\lambda 2$.
\end{proof}

\begin{remark}[An analytical lower bound of the threshold of the UF decoder under the circuit-level local stochastic error model]
We estimate a lower bound on the threshold $p_\mathrm{th}$ for the UF decoder.
We take the constants $\beta, \gamma, \lambda$
satisfying
\begin{align}
    &{\beta\over \gamma} = {1-4\lambda^{-2\log 2} \over 3}, \quad \gamma = 2,\\
    &{1-4\lambda^{-2\log 2} \over 3} = 6[\zeta(\log \lambda)-1+\lambda^{-2\log 2}] > 2\lambda^{-2\log 2}.
\end{align}
For reference, these values are
\begin{align}
    \beta \approx 0.66, \quad \gamma = 2, \quad \lambda \approx 93.
\end{align}
Substituting these parameters in Eq.~\eqref{eq:analytical_threshold} yields
\begin{align}
\label{eq:analytical_lower_bound_threshold}
    p_\mathrm{th} > 6.6\times 10^{-25}/\xi,
\end{align}
where $\xi$ is the constant in Eq.~\eqref{eq:tilde_p}, given by the maximum number of circuit locations whose faults can flip a stabilizer measurement outcome.
For the surface code, it typically holds that
\begin{align}
\label{eq:xi_bound}
    \xi \le 10,
\end{align}
as is the case, for example, for the syndrome extraction circuit in Supplementary Information~\ref{appendix:circuit}.
Under the bound~\eqref{eq:xi_bound}, Eq.~\eqref{eq:analytical_lower_bound_threshold} implies an analytical lower bound of the threshold
\begin{align}
    p_\mathrm{th} > 6.6 \times 10^{-26}.
\end{align}
Previously, numerical evidence, e.g., in Ref.~\cite{Chan2023actisstrictlylocal}, has suggested that the UF decoder could exhibit a threshold behavior under a circuit-level error model with numerical threshold estimate $p_\mathrm{th} \approx10^{-2}\sim 10^{-3}$, while the exact threshold value may depend on the details on parameters in the error model; by contrast, our key contribution is to rigorously prove a nonzero lower bound of $p_\mathrm{th}$ for the UF decoder.
\end{remark}

\paragraph*{Proof of quasi-polylogarithmic average runtime of parallel UF decoders}
Using the threshold theorem for the UF decoder, we prove that the average runtime of the parallel UF decoder is quasi-polylogarithmic in the size of the surface code.

\begin{proof}[Proof of Theorem~\ref{thm:parallel_UF_runtime}]
As shown in the main text,  level-$k$ extended UF clusters can be decoded in parallel in time $O(\poly(d_k)) \leq \lambda^{\tau k \log k}$ for a sufficiently large constant $\tau$, where we use Eq.~\eqref{eq:d_k_uf}.
The probability of having a level-$k$ extended UF cluster in the decoding window is upper-bounded by
\begin{align}
    \min\left\{\abs{E} \cdot p_k,1 \right\}
    &= \min\left\{O(d^3) \cdot O\qty((p/p_\mathrm{th})^{2^k}),1 \right\},
\end{align}
where $1$ is a trivial upper bound.
For the code distance $d$, we define
\begin{align}
    \bar{k}\coloneqq \max\left\{k \middle| d^3 (p/p_\mathrm{th})^{2^k} \geq 1\right\} = O(\log \log d).
\end{align}
Then, the expected runtime of the parallel UF decoder, averaged over possible error configurations in the $O(d^3)$-sized window, is upper-bounded by
\begin{align}
    &\sum_{k=1}^{\infty} \lambda^{\tau k \log k} \cdot \min\left\{\abs{E} \cdot p_k,1\right\}\nonumber\\
    &\leq \sum_{k=1}^{\bar{k}} \lambda^{\tau k \log k} \nonumber\\
    &+ \sum_{k=\bar{k}+1}^{\infty} \lambda^{\tau k \log k} \cdot O(d^3) \cdot O\qty((p/p_\mathrm{th})^{2^k}).
\end{align}
The first term is upper-bounded by
\begin{align}
    \bar{k} \lambda^{\tau \bar{k} \log \bar{k}} = \exp\qty[O(\log\log d\cdot\log \log \log d)].
\end{align}
To upper-bound the second term, we consider the ratio test on a sequence $\qty{a_k \coloneqq \lambda^{\tau k \log k} (p/p_\mathrm{th})^{2^k}}_{k=1,2,\ldots}$.
Since $\lim_{k\to\infty} a_{k+1}/a_k = 0$ holds, for any constant $\epsilon>0$, there exists $\bar{k}'$ such that,  for any $k\geq \bar{k}'$, it holds that $a_{k+1}/a_{k}<1-\epsilon$.
Since $\bar{k}'$ does not depend on $d$, we can choose $d$ large enough so that $\bar{k} \geq \bar{k}'$ holds.
Then, the second term is upper-bounded by
\begin{align}
    &\sum_{k=\bar{k}+1}^{\infty} O(d^3) \cdot O(a_k)\nonumber\\
    &\leq  O(d^3) \cdot O(a_{\bar{k}}) \cdot \sum_{k=\bar{k}+1}^{\infty} (1-\epsilon)^{k-\bar{k}}\\
    &= O\qty(d^3 (p/p_\mathrm{th})^{2^{\bar{k}}} \lambda^{\tau \bar{k} \log \bar{k}}) \cdot {1-\epsilon \over \epsilon}\\
    &= \exp\qty[O(\log\log d\cdot \log \log \log d)].
\end{align}
Therefore, the expected runtime of the parallel UF decoder is upper-bounded by
\begin{align}
&\exp\qty[O(\log\log d\cdot \log \log \log d)]\notag\\
&=\exp\qty[O(\log\log n\cdot \log \log \log n)].
\end{align}
\end{proof}

\paragraph*{Proof of the threshold theorem under the greedy decoder}
Similarly to the UF decoder, we can show the threshold theorem under the greedy decoder as follows.
\begin{proof}[Proof of Theorem~\ref{thm:threshold_theorem_greedy}]
We assume
\begin{align}
    \label{eq:constraint_bkdk_greedy}
    d_k\geq 1, \quad b_k-1 > d_k+1.
\end{align}
Since $d_k\geq 1$ holds, any adjacent edges $e, e'\in E$ are in the same cluster.
Thus, each cluster does not separate any connected component of the error locations.
Since the maximum distance between two active detectors in the same level-$k$ cluster is upper-bounded by
\begin{align}
    d_k+1
\end{align}
and the minimum distance between two different level-$k$ clusters is lower bounded by
\begin{align}
    b_k-1,
\end{align}
the greedy algorithm always finds pairs of active detectors within the same cluster under the constraint in Eq.~\eqref{eq:constraint_bkdk_greedy}.

We take the sequences $\qty{(b_k, d_k)}_{k=0,1,\ldots}$ as
\begin{align}
    b_k = \beta \lambda^{k+1}+1, \quad d_k = \gamma \lambda^{k}-1,
\end{align}
which corresponds to $f(k) = k$ in Eq.~\eqref{eq:bkdk_scaling}, satisfying Eq.~\eqref{eq:f_convergence} for $c=3$.
Then, the constraints in Eqs.~\eqref{eq:bkdk_constraint} and \eqref{eq:constraint_bkdk_greedy} are satisfied if
\begin{align}
    \lambda\geq 2, \quad \lambda^{-1}<{\beta\over \gamma}\leq {1-4\lambda^{-1} \over 3}.
\end{align}
Note that the constants $\beta, \gamma, \lambda$ satisfying the above constraints exist since $\lambda^{-1}<{1-4\lambda^{-1} \over 3}$ holds for $\lambda>7$.
If
\begin{align}
    k\leq k_0\coloneqq \max\{k \mid d_k+1 < d\} = \log_\lambda d + O(1)
\end{align}
holds, the diameter of the level-$k$ cluster cannot be larger than $d$, and no logical error occurs.
Thus, the logical error rate $p_\mathrm{L}$ is upper bounded by
\begin{align}
    p_\mathrm{L}&\leq \prob[\sfN_{k_0+1}\neq\varnothing]\\
    &\leq \sum_{e\in E} \prob[e\in \sfN_{k_0+1}]\\
    &\leq \abs{E} \cdot p_{k_0+1}\\
    &= O\left(d^3\right) \cdot  (p/p_\mathrm{th})^{2^{k_0+1}}\\
    &= (p/p_\mathrm{th})^{\Omega(d^\eta)},
    \label{eq:greedy_decoder_threshold}
\end{align}
where $\eta$ is defined by $\eta \coloneqq \log_\lambda 2$.
\end{proof}

\paragraph*{Effective distance of the greedy decoder}
As indicated by the theorem below, the greedy decoder has an effective distance of $O\qty(d^{0.6309})$; equivalently, its logical error rate cannot scale better than $p_\mathrm{L}=(p/p_\mathrm{th})^{O\qty(d^{0.6309})}$ in $d$ (see also Ref.~\cite{lake2025local} for a similar argument on the effective distance of the CA decoder).
While similar ideas have been discussed in Refs.~\cite{e17041946,hutter2015improved}, we provide a complete derivation of this bound tailored to our setting under the circuit-level local stochastic error model.

\begin{theorem}[Error configuration limiting effective distance of the greedy decoder]
    For the $[[d^2,1,d]]$ surface code, there exists a set of error locations in the syndrome extraction circuit such that the greedy decoder causes a logical error with the number $N$ of error locations given by
    \begin{align}
    \label{eq:effective_distance}
        N\leq 2^{\log_3 {108\over 13}} \left(d-{3\over 4}\right)^{\log_3 2} \approx 3.803 \left(d-{3\over 4}\right)^{0.6309}.
    \end{align}
\end{theorem}
\begin{proof}
We provide a set of $N$ error locations such that the greedy decoder causes a logical error, inspired by Fig.~1 of Ref.~\cite{reingold1981greedy}.
We take a one-dimensional chain in the detector graph that connects boundary vertices with $d$ edges, where the shortest path between any two nodes on the chain is also included in the chain.
Then, the greedy decoder pairs active detectors along this chain.
We represent the chain as a subset $S$ of edges of the detector graph with chain length $\abs{S} = d$, and decompose $S$ into three segments with lengths
\begin{align}
    \left\lfloor{\abs{S}-\left\lfloor {\abs{S}-2\over 3} \right\rfloor \over 2}\right\rfloor, \left\lfloor {\abs{S}-2\over 3} \right\rfloor, \left\lceil{\abs{S}-\left\lfloor {\abs{S}-2\over 3} \right\rfloor \over 2}\right\rceil.
\end{align}
Then, we remove the middle segment to obtain two disjoint segments $S_1$ and $S_2$.
We recursively apply the same removal procedure on $S_1$ and $S_2$ until the length of the segments becomes less than $5$ to obtain the set of disjoint segments $\{S_i\}_{i=1}^{p}$, where $p$ is the number of the disjoint segments.

As shown below, the greedy decoder causes a logical error for a set of errors corresponding to $\{S_i\}_{i=1}^{p}$, and the number of errors $N = \sum_{i=1}^{p} \abs{S_i}$ satisfies Eq.~\eqref{eq:effective_distance}.
Since
\begin{align}
    \left\lfloor {\abs{S}-2\over 3} \right\rfloor < \left\lfloor{\abs{S}-\left\lfloor {\abs{S}-2\over 3} \right\rfloor \over 2}\right\rfloor\leq \left\lceil{\abs{S}-\left\lfloor {\abs{S}-2\over 3} \right\rfloor \over 2}\right\rceil
\end{align}
holds, the greedy decoder always matches the middle segment in the decomposition of $S$ (see Fig.~\ref{fig:failure_pattern} for an example of $d=23$).
Therefore, the active-detector pair found by the greedy decoder is the complement of the error locations in the chain, leading to a logical error.

The number of errors $N$ is upper-bounded as follows.
Since $\abs{S_i}\leq 4$ holds, $N\leq 4p$ holds.
In each decomposition of the segment, the right-most segment is the longest.
Thus, if the right-most segment is decomposed $k_0$ times, $p$ satisfies $p\leq 2^{k_0}$.
The length of the right-most segment after $k$ decompositions, denoted by $l_k$, is given by
\begin{align}
    l_0 &= d,\\
    l_{k+1}&= \left\lceil{l_k-\left\lfloor {l_k-2\over 3} \right\rfloor \over 2}\right\rceil \leq {l_k\over 3} + {7\over 6},
\end{align}
where we use the following inequalities for any integer $m$ and positive integer $n$:
\begin{align}
    \left\lfloor {m\over n} \right\rfloor \geq {m-n+1 \over n}, \quad \left\lceil {m\over n} \right\rceil \leq {m+n-1 \over n}.
\end{align}
Thus, we obtain
\begin{align}
    l_k\leq \left(d-{3\over 4}\right)\left(1\over 3\right)^k + {7\over 4}.
\end{align}
Since $k_0$ is given by
\begin{align}
    k_0 = \min_{l_k< 5} k, 
\end{align}
we obtain
\begin{align}
    k_0 \leq \log_3\left[{4\over 13}\left(d-{3\over 4}\right)\right]+1.
\end{align}
Therefore, we obtain
\begin{align}
    N&\leq 2^{k_0+2}\\
    &\leq 2^{\log_3\left[{108\over 13}\left(d-{3\over 4}\right)\right]}\\
    &= 2^{\log_3 {108\over 13}} \left(d-{3\over 4}\right)^{\log_3 2}.
\end{align}
\end{proof}

\begin{figure}
    \centering
    \includegraphics[width=3.4in]{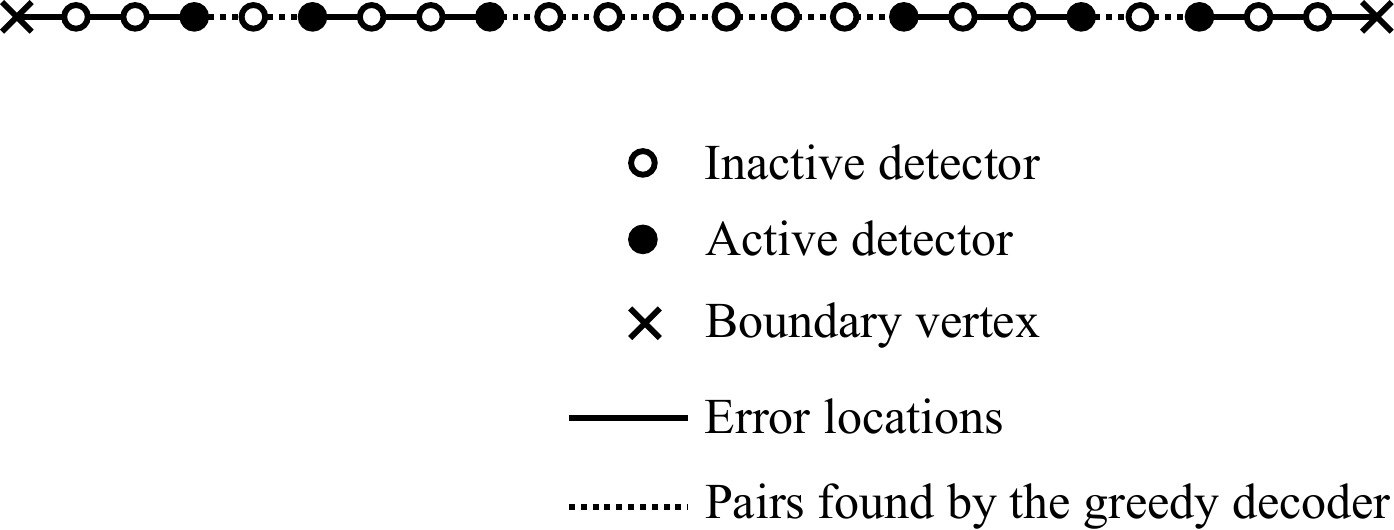}
    \caption{
        An example of error locations where the greedy decoder causes a logical error for the distance $d=23$.}
    \label{fig:failure_pattern}
\end{figure}

\section*{Data availability}

No data is used in this study.

\section*{Code availability}

No code is used in this study.

\bibliography{main}

\section*{Acknowledgments}
EL and HY discussed part of this work during the KIAS workshop on Quantum Information Dynamics and Non-Equilibrium Physics (Quantum Information and Dynamics 2025).
HY acknowledges Nicolas Delfosse and Shouzhen Gu for discussions at the $7$th International Conference on Quantum Error Correction (QEC25), and Naomi Nickerson for discussions at the Seeking Quantum Advantage Workshop and Conference (SEEQA2025). EL thanks Nathaniel Selub for discussions and collaboration on related work.
We thank Ryo Mikami for pointing out typographical errors in the manuscript.
SY was supported by Japan Society for the Promotion of Science (JSPS) KAKENHI Grant Number 23KJ0734, FoPM, WINGS Program, the University of Tokyo, and DAIKIN Fellowship Program, the University of Tokyo.
EL was supported by a Miller Fellowship. 
HY was supported by JST PRESTO Grant Number JPMJPR201A, JPMJPR23FC, JSPS KAKENHI Grant Number JP23K19970, JST CREST Grant Number JPMJCR25I5, JST [Moonshot R\&D] [Grant Number JPMJMS256J], and Faculty Research Funding from Google Quantum AI\@. 

\section*{Author contributions}
SY, EL, and HY contributed to the conception, analysis, and interpretation of the work, and to the preparation and revision of the manuscript.

\section*{Competing Interests}
The authors declare no competing interests.

\clearpage

\renewcommand{\theequation}{S\arabic{equation}}
\renewcommand{\theHequation}{S\arabic{equation}}
\renewcommand{\thetheorem}{S\arabic{theorem}}
\renewcommand{\theHtheorem}{S\arabic{theorem}}
\renewcommand{\thesection}{S\arabic{section}}
\renewcommand{\theHsection}{S\arabic{section}}
\renewcommand{\thetable}{S\arabic{table}}
\renewcommand{\theHtable}{S\arabic{table}}
\renewcommand{\thefigure}{S\arabic{figure}}
\renewcommand{\theHfigure}{S\arabic{figure}}

\setcounter{equation}{0}
\setcounter{theorem}{0}
\setcounter{figure}{0}
\setcounter{section}{0}

\onecolumngrid

\section*{Supplementary Information}

\section{Syndrome extraction circuit for the surface code}
\label{appendix:circuit}

The rotated surface code is a $[[d^2, 1, d]]$ Calderbank-Shor-Steane (CSS) code.
For simplicity, we present the case of an odd distance $d$ in this work, but all the arguments can be extended to the case when $d$ is even.
Each qubit is labeled by $\vec{r} = (x, y)\in \{1,\ldots,d\}^2$, and the stabilizer generators of the surface code are given by
\begin{align}
\begin{cases}
    \prod_{\vec{r}\in f_X} X_{\vec{r}}  & \forall f_X\in F_X\\
    \prod_{\vec{r}\in f_Z} Z_{\vec{r}} & \forall f_Z\in F_Z
\end{cases},
\end{align}
where $F_X$ and $F_Z$ are the sets of faces given by
\begin{align}
\label{eq:definition_face}
\begin{split}
    F_X&\coloneqq F_X^\mathrm{bulk} \sqcup F_X^\mathrm{boundary},\\
    F_Z&\coloneqq F_Z^\mathrm{bulk} \sqcup F_Z^\mathrm{boundary},\\
    F_X^\mathrm{boundary}&\coloneqq F_X^{x=0}\sqcup F_X^{x=d-1},\\
    F_Z^\mathrm{boundary}&\coloneqq F_Z^{y=0} \sqcup F_Z^{y=d-1},\\
    F_X^{\mathrm{bulk}}&\coloneqq \left\{S(x,y) \middle| (x,y)\in [d-1]^2, x+y = 1 \mod 2 \right\},\\
    F_Z^{\mathrm{bulk}}&\coloneqq \left\{S(x,y) \middle| (x,y)\in [d-1]^2, x+y = 0 \mod 2 \right\},\\
    F_X^{x=0}&\coloneqq\left\{D_y(0, 2i) \middle| i\in \left[{d-1\over 2}\right]\right\},\\
    F_X^{x=d-1}&\coloneqq\left\{D_y(d-1, 2i+1)  \middle| i\in \left[{d-1\over 2}\right]\right\},\\
    F_Z^{y=0}&\coloneqq\left\{D_x(2i+1, 0) \middle| i\in \left[{d-1\over 2}\right]\right\},\\
    F_Z^{y=d-1}&\coloneqq \left\{D_x(2i, d-1) \middle| i\in \left[{d-1\over 2}\right]\right\},
\end{split}
\end{align}
$S(x,y)$ is the square defined by
\begin{align}
    S(x,y)\coloneqq \{(x,y), (x+1, y), (x,y+1), (x+1, y+1)\},
\end{align}
and $D_x(x,y)$ and $D_y(x,y)$ are the digons defined by
\begin{align}
    D_x(x,y)\coloneqq \{(x,y), (x+1, y)\},\quad
    D_y(x,y)\coloneqq \{(x,y), (x, y+1)\}.
\end{align}
See Fig.~\ref{fig:surface_code_stabilizer} in the main text for the example of $d=3$.

For concreteness, we present the syndrome extraction circuit shown in Ref.~\cite{mcewen2023relaxing}, while our analysis applies to a general class of circuits satisfying the spacetime locality condition of the detector graph shown in Methods.
In the syndrome extraction circuit, we add auxiliary qubits, called the measurement qubits, labeled by a face $f\in F$.
We assign the coordinates to the measurement qubits corresponding to $f\in F$ by
\begin{align}
    \vec{r}_f\coloneqq
    \begin{cases}
        \left(x+{1\over 2}, y+{1\over 2}\right) & (f\in F_X^\mathrm{bulk} \sqcup F_Z^\mathrm{bulk})\\
        \left(-{1\over 2}, 2i+{1\over 2}\right) & (f\in F_X^{x=0})\\
        \left(d-{1\over 2}, 2i+{3\over 2}\right) & (f\in F_X^{x=d-1})\\
        \left(2i+{3\over 2}, -{1\over 2}\right) & (f\in F_Z^{y=0})\\
        \left(2i+{1\over 2}, d-{1\over 2}\right) & (f\in F_Z^{y=d-1})
    \end{cases},
\end{align}
where $x, y, i$ are defined in Eq.~\eqref{eq:definition_face}.
The syndrome extraction circuit consists of the reset step, the 4 CNOT steps, and the measurement step, given as follows (see also Fig.~\ref{fig:syndrome_extraction_circuit}):
\begin{enumerate}
    \item Reset step: Reset the states of the measurement qubits $\vec{r}_{f_X}$ in $\ket{+}$ and $\vec{r}_{f_Z}$ in $\ket{0}$ for all $f_X\in F_X$ and $f_Z\in F_Z$.
    \item CNOT steps: Apply the CNOT gates in the following ordering:
    \begin{align}
    \begin{cases}
        \mathrm{CNOT}_{\vec{r}_{f_X}, \vec{r}_{f_X}+\vec{e}_{LT}} & \forall f_X\in F_X^\mathrm{bulk} \sqcup F_X^{x=d-1}\\
        \mathrm{CNOT}_{\vec{r}_{f_Z}+\vec{e}_{LT}, \vec{r}_{f_X}} & \forall f_Z\in F_Z^\mathrm{bulk}\sqcup F_Z^{y=0}
    \end{cases},\\
    \begin{cases}
        \mathrm{CNOT}_{\vec{r}_{f_X}, \vec{r}_{f_X}+\vec{e}_{LB}} & \forall f_X\in F_X^\mathrm{bulk} \sqcup F_X^{x=d-1}\\
        \mathrm{CNOT}_{\vec{r}_{f_Z}+\vec{e}_{RT}, \vec{r}_{f_X}} & \forall f_Z\in F_Z^\mathrm{bulk}\sqcup F_Z^{y=0}
    \end{cases},\\
    \begin{cases}
        \mathrm{CNOT}_{\vec{r}_{f_X}, \vec{r}_{f_X}+\vec{e}_{RT}} & \forall f_X\in F_X^\mathrm{bulk} \sqcup F_X^{x=0}\\
        \mathrm{CNOT}_{\vec{r}_{f_Z}+\vec{e}_{LB}, \vec{r}_{f_X}} & \forall f_Z\in F_Z^\mathrm{bulk}\sqcup F_Z^{y=d-1}
    \end{cases},\\
    \begin{cases}
        \mathrm{CNOT}_{\vec{r}_{f_X}, \vec{r}_{f_X}+\vec{e}_{RB}} & \forall f_X\in F_X^\mathrm{bulk} \sqcup F_X^{x=0}\\
        \mathrm{CNOT}_{\vec{r}_{f_Z}+\vec{e}_{RB}, \vec{r}_{f_X}} & \forall f_Z\in F_Z^\mathrm{bulk}\sqcup F_Z^{y=d-1}
    \end{cases},
    \end{align}
    where $\vec{e}_{LT}, \vec{e}_{LB}, \vec{e}_{RT}, \vec{e}_{RB}$ are defined by
    \begin{align}
        \vec{e}_{LT}\coloneqq \left(-{1\over 2}, {1\over 2}\right),\quad 
        \vec{e}_{LB}\coloneqq \left(-{1\over 2}, -{1\over 2}\right),\quad
        \vec{e}_{RT}\coloneqq \left({1\over 2}, {1\over 2}\right),\quad
        \vec{e}_{RB}\coloneqq \left({1\over 2}, -{1\over 2}\right),
    \end{align}
    and $LT$, $LB$, $RT$, and $RB$ stand for left-top, left-bottom, right-top, and right-bottom, respectively.
    \item Measure the measurement qubits $\vec{r}_{f_X}$ in $X$ basis and $\vec{r}_{f_Z}$ in $Z$ basis for all $f_X\in F_X$ and $f_Z\in F_Z$,
\end{enumerate}

\begin{figure}
    \centering
    \includegraphics[width=7.0in]{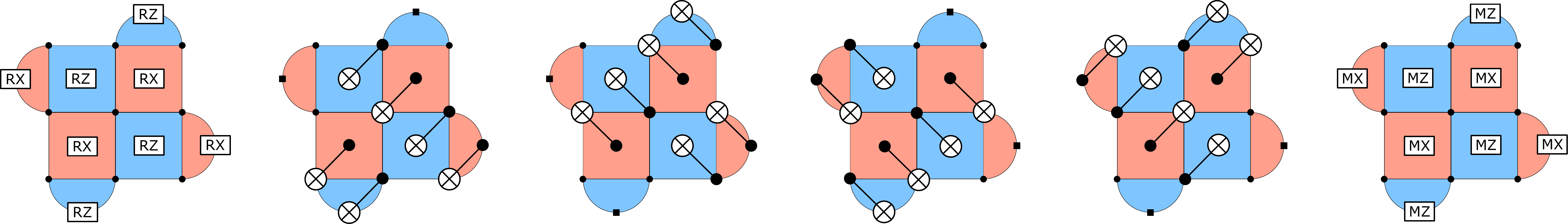}
    \caption{Syndrome extraction circuit for the surface code with distance $d=3$, where time flows from left to right.
    RX and RZ represent the reset operations to $\ket{+}$ and $\ket{0}$, respectively.
    MX and MZ represent the measurements in the $X$ and $Z$ bases, respectively.}
    \label{fig:syndrome_extraction_circuit}
\end{figure}

We repeat the syndrome extraction circuit for $d$ times, and we write the measurement outcome of the measurement qubit $\vec{r}_f$ in the $i$-th syndrome extraction circuit as
\begin{align}
    m_{f, i} \in \{0, 1\}
\end{align}
for $f\in F_X\sqcup F_Z$ and $i\in [d]$.
We define the detectors $\hat{m}_{f,i} \in \{0, 1\}$ by
\begin{align}
    \hat{m}_{f,i} \coloneqq \begin{cases}
        m_{f,0} & (i=0)\\
        m_{f,i-1} + m_{f,i} \mod 2 & (i>0)
    \end{cases}.
\end{align}
We assign the spacetime coordinate $(\vec{r}_f, i)$ to the detector $\hat{m}_{f,i}$, and the detector at $(\vec{r}_f, i)$ is called an active detector if $\hat{m}_{f,i}=1$.

The detector graph corresponding to this syndrome extraction circuit is shown in, e.g., Ref.~\cite{stim_tutorial}, and it satisfies the following properties:
\begin{itemize}
\item the degree of the detector graph is at most $12$;
\item the length of each edge is at most $\sqrt{3}$ (i.e., $C=\sqrt{3}$ in Theorem~\ref{thm:spacetime_locality}).
\end{itemize}

Then, we can show the following lemma.
Note that in our analysis, we do not aim at obtaining the tightest bound, and one may improve a lower bound on $\Lambda$ in this lemma to obtain a better analytical lower bound of the threshold.

\begin{lemma}[Parameters of the syndrome extraction circuit]
    The detector graph $G$ corresponding to this syndrome extraction circuit satisfies Eq.~\eqref{eq:upper_bound_on_neighboring_edges} for
    \begin{align}
        \Lambda = 48\sqrt{3} \pi, \quad \Delta = 3.
    \end{align}
\end{lemma}
\begin{proof}
    Suppose a vertex $v$ is incident with an edge $e$, and the distance of two edges $e, e'\in E$ is bounded by $\dist(e,e')\leq r$.
    Then, there exists a vertex $v'$ incident with $e'$ such that the distance between $v$ and $v'$ is bounded by $\dist(v,v')\leq r$.
    Since each vertex is incident with at most $12$ edges, we obtain
    \begin{align}
        \abs{B_e(r)} \leq 12 \abs{B_{v}(r)},
    \end{align}
    where $B_v(r)$ is the set of vertices within distance $r$ from $v$.
    Since the length of each edge in the detector graph is at most $\sqrt{3}$, we have
    \begin{align}
        \abs{B_v(r)}\leq {4\pi \over 3}(\sqrt{3}r)^3 = 4\sqrt{3}\pi r^3.
    \end{align}
    Thus, we obtain
    \begin{align}
        \abs{B_e(r)}\leq 48 \sqrt{3} \pi r^3.
    \end{align}
\end{proof}

\section{Proof of Lemma~\ref{lem:sparsity}}

We summarize the definitions on the graph used in the proof.
Throughout, a graph refers to a simple undirected graph.

\begin{definition}[Distance in graphs]
    Let $G = (V,E)$ be a graph.
    \begin{enumerate}
        \item The distance $\dist(v, v')$ between two vertices $v, v'\in V$ is defined by the length of the shortest path connecting $v$ and $v'$.
        \item The line graph $L(G)$ is defined by a graph whose set of vertices is given by $E$, and the set of edges is given by
        \begin{align}
            \{(e, e')\in E^2 \mid \text{$e$ and $e'$ are adjacent in the graph $G$}\}.
        \end{align}
        \item The distance $\dist(e, e')$ between two edges $e, e'\in E$ is defined by the distance between $e$ and $e'$ in the line graph $L(G)$.
        \item The distance between two subsets of edges $C, C' \subset E$ is defined by
        \begin{align}
            \dist(C, C')\coloneqq \min_{e\in C, e'\in C'} \dist(e,e').
        \end{align}
        \item The set of edges within a distance $r$ of an edge $e$ is denoted by $B_e(r)$, i.e.,
        \begin{align}
            B_e(r)\coloneqq \{e'\mid \dist(e,e')\leq r\}.
        \end{align}
        \item The diameter of a subset $C\subset E$ of edges is defined by
        \begin{align}
            \diam(C) \coloneqq \max_{e, e'\in C} \dist(e,e').
        \end{align}
    \end{enumerate}
\end{definition}

\begin{figure}[t]
    \centering
    \includegraphics[width=3.4in]{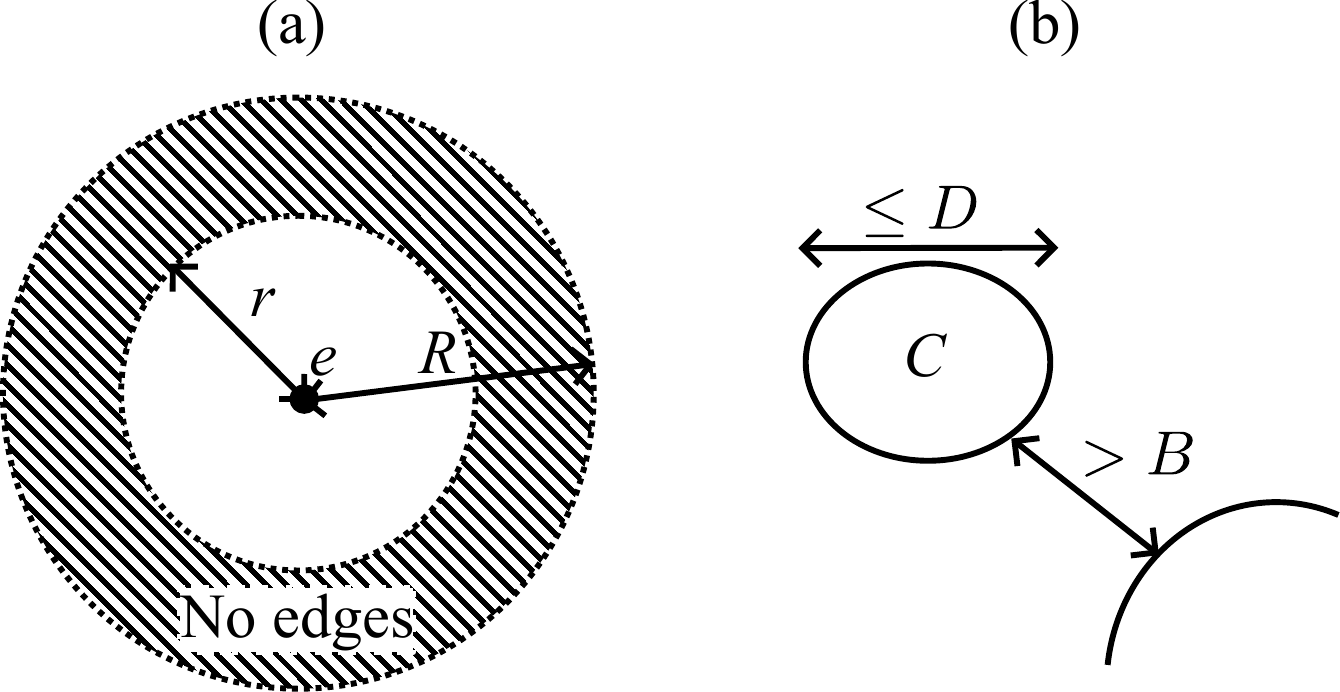}
    \caption{Illustration of isolation and clustering in the graph.
    (a) The edge $e$ is $(r,R)$-isolated if there is no other edge in the annulus $B_e(R)\setminus B_e(r)$.
    (b) The set $C$ is a $(D,B)$-cluster in $N$ if the diameter of $C$ is at most $D$ and the distance between $C$ and $N\setminus C$ is more than $B$.}
    \label{fig:isolation_and_cluster}
\end{figure}

\begin{definition}[Isolation and clustering of edges in graphs (see also Fig.~\ref{fig:isolation_and_cluster})]
    Let $G = (V,E)$ be a graph, and $N\subset E$ be a subset of edges.
    \begin{enumerate}
        \item Two edges $e, e'\in N$ are called $(r,R)$-isolated if $e'\notin B_e(R)\setminus B_e(r)$, and otherwise called $(r,R)$-linked.
        \item An edge $e\in N$ is called $(r,R)$-isolated in $N$ if $e$ and $e'$ are $(r,R)$-isolated with all $e'\in N$.
        \item The set $N$ is called $(r,R)$-isolated if each $e\in N$ is $(r,R)$-isolated in $N$.
        \item A subset $C\subset N$ is called a $(D,B)$-cluster in $N$ if
        \begin{align}
            \diam(C)\leq D,\\
            \dist(N\setminus C, C)> B.
        \end{align}
        Any element $e\in C$ of a $(D,B)$-cluster in $N$ is called $(D,B)$-clustered in $N$.
    \end{enumerate}
\end{definition}

\begin{proposition}[Relation between isolation and clustering]
\label{prop:isolated_implies_clustered}
    Let $G = (V,E)$ be a graph and $N, N'$ be subsets of edges satisfying $N\subset N'\subset E$.
    If $e$ is $(r,R)$-isolated in $N'$, $e$ is $(2r, R-r)$-clustered in $N$.
\end{proposition}
\begin{proof}
    We define a set of edges by
    \begin{align}
        C\coloneqq B_e(r) \cap N'.
    \end{align}
    Since $e$ is $(r,R)$-isolated in $N'$, $(N\setminus C) \cap B_e(R) \subset (N'\setminus C) \cap B_e(R) = \varnothing$ holds.
    Thus, $\dist(e,e')>R$ holds for any $e'\in N\setminus C$.
    Since $\dist(e,e'')\leq r$ holds for any $e''\in C$ by definition, we obtain
    \begin{align}
        \dist(N\setminus C, C)
        &= \min_{e'\in N\setminus C, e''\in C} \dist(e',e'')\\
        &\geq \min_{e'\in N\setminus C, e''\in C} [\dist(e,e')-\dist(e,e'')]\\
        &> R-r,\\
        \diam(C)
        &= \max_{e', e''\in C} \dist(e',e'')\\
        &\leq \max_{e',e''\in C} [\dist(e,e')+\dist(e,e'')]\\
        &\leq 2r,
    \end{align}
    where we use the triangle inequality for distances.
\end{proof}

\begin{definition}[Clustered sets and isolated sets]
\label{def:cluster}
    Let $G = (V,E)$ be a graph and $N\subset E$ be a subset of edges.
    We consider monotonically increasing sequences $\{d_k\}_{k=0, 1, \ldots}$ and $\{ b_k\}_{k=0, 1, \ldots}$ satisfying
    \begin{align}
        b_k \geq d_k >0.
    \end{align}
    We inductively define the level-$k$ residual sets $\sfN_k$ and $\mcN_k$ of unclustered and non-isolated edges, respectively, as follows.
    \begin{align}
        \sfN_0 &\coloneqq \mcN_0 \coloneqq N,\\
        \sfN_{k+1} &\coloneqq \sfN_k \setminus \{e\in \sfN_k \mid \text{$e$ is $(d_k, b_k)$-clustered in $\sfN_k$}\},\\
        \mcN_{k+1} &\coloneqq \mcN_k \setminus \{e\in \mcN_k \mid \text{$e$ is $(d_k/2, b_k+d_k/2)$-isolated in $\mcN_k$}\}.
    \end{align}
    Using these sets, we define the level-$k$ clustered edge set $\sfE_k$ and level-$k$ isolated edge set $\mcE_k$ by
    \begin{align}
        \sfE_k&\coloneqq \sfN_k\setminus \sfN_{k+1} = \{e\in \sfN_k \mid \text{$e$ is $(d_k, b_k)$-clustered in $\sfN_k$}\},\\
        \mcE_k&\coloneqq \mcN_k\setminus \mcN_{k+1} = \{e\in \mcN_k \mid \text{$e$ is $(d_k/2, b_k+d_k/2)$-isolated in $\mcN_k$}\}.
    \end{align}
\end{definition}

The following proposition shows that the set $\sfE_k$ can be decomposed into disjoint $(d_k, b_k)$-clusters in $\sfN_k$.

\begin{proposition}[Separation of clusters]
    \label{prop:disjoint_clusters}
    Let $G = (V,E)$ be a graph, $N\subset E$ be a subset of edges, and $C$ and $C'$ are $(D,B)$-clusters in $N$.
    If $B\geq D$ and $C\neq C'$, then $C\cap C' = \varnothing$ holds.
\end{proposition}

\begin{proof}
    We prove this proposition by assuming $B\geq D$, $C\neq C'$, and $C\cap C' \neq \varnothing$ to derive a contradiction.
    Suppose that $e\in C\cap C'$ and $e'\in C\setminus C'$.
    Since $\diam(C)\leq D$ and $e, e'\in C$ hold, we obtain
    \begin{align}
        \dist(e,e')\leq \diam(C)\leq D\leq B.
    \end{align}
    On the other hand, since $\dist(N\setminus C', C')\geq B$ and $e\in C', e'\in N\setminus C'$ hold, we obtain
    \begin{align}
        \dist(e,e')\geq \dist(N\setminus C', C')> B,
    \end{align}
    which shows a contradiction.
\end{proof}

\begin{proposition}[Inclusion relation between clustered sets and isolated sets]
    \label{prop:Nk_inclusion}
    For all $k$, it holds that $\sfN_k \subset \mcN_k$.
\end{proposition}
\begin{proof}
    We show the statement by induction on $k$.
    Since $\sfN_0 = \mcN_0 = N$, $\sfN_0 \subset \mcN_0$ holds.
    Suppose $\sfN_k\subset \mcN_k$ holds.
    For any $e\in \mcE_k$, since $e$ is $(d_k/2, b_k+d_k/2)$-isolated in $\mcN_k$, $e$ is $(d_k, b_k)$-clustered in $\sfN_k$ (Proposition~\ref{prop:isolated_implies_clustered}), i.e., $e\in \sfE_k$.
    Thus, we obtain $\mcE_k\subset \sfE_k$, i.e., $\sfN_{k+1}\subset \mcN_{k+1}$.
\end{proof}

\begin{lemma}[Cardinality bound on the set of minimal-size subsets of edges for isolated sets]
    \label{lem:min_size_set}
    Let $G = (V,E)$ be a graph.
    For $e\in E$, we let $M_k(e)$ denote the set of minimal-size subsets of $E$ implying $e\in \mcN_k$, i.e., each subset $N\subset E$ is in $M_k(e)$ if and only if
    \begin{enumerate}
        \item $e\in \mcN_k\subset N$, and
        \item for any proper subset $N'\subsetneq N$, $e\notin\mcN'_k\subset N'$ holds, where $\mcN'_k$ is the level-$k$ isolated set in Definition~\ref{def:cluster} with $N$ substituted with $N'$.
    \end{enumerate}
    If
    \begin{align}
        d_{k+1}\geq 4d_k + 3b_k
    \end{align}
    holds for all $k\geq 0$, the following relations hold:
    \begin{align}
    \label{eq:radius_bound_N}
        N&\subset B_e(d_k/4) \quad \text{for all $N\in M_k(e)$},\\
    \label{eq:cardinality_of_N}
        \abs{N} &= 2^k \quad \text{for all $N\in M_k(e)$}.
    \end{align}
    In addition, if the graph $G$ satisfies,  for some constants $\Delta, \Lambda>0$,
    \begin{align}
        \abs{B_e(r)} \leq \Lambda r^{\Delta} \quad \text{for all $e\in E$ and $r>0$},
    \end{align}
    then the cardinality of $M_k(e)$ is upper bounded by
    \begin{align}
    \label{eq:cardinality_Mk}
        \abs{M_k(e)} &\leq \prod_{j=0}^{k-1} \left[\Lambda \qty(b_{k-j-1}+\frac{d_{k-j-1}}{2})^\Delta\right]^{2^{j}}.
    \end{align}
\end{lemma}
\begin{proof}
    We will first show Eq.~\eqref{eq:radius_bound_N} by induction on $k$.
    Since $M_0(e) = \{\{e\}\}$, Eq.~\eqref{eq:radius_bound_N} holds for $k=0$.
    Assuming that Eq.~\eqref{eq:radius_bound_N} holds for $k$, we will show Eq.~\eqref{eq:radius_bound_N} for $k+1$.
    If $e\in \mcN_{k+1}$ holds, then by definition, there exists $e'\in \mcN_{k}$ such that $e$ and $e'$ are $(d_k/2, b_k+d_k/2)$-linked.
    By definition of $M_{k+1}(e)$, any $N\in M_{k+1}(e)$ includes, as subsets, some $N_1\in M_{k}(e)$ and some $N_2\in M_{k}(e')$ for an edge $e'$ such that $e$ and $e'$ are $(d_k/2, b_k+d_k/2)$-linked; moreover, by minimality, $N$ contains no elements other than those in $N_1$ and $N_2$.
    Thus, we obtain
    \begin{align}
    \label{eq:decomposition_N}
        N = N_1 \cup N_2.
    \end{align}
    By assumption, we have $N_1\subset B_e(d_k/4)$ and $N_2\subset B_{e'}(d_k/4)$.
    Therefore, we obtain
    \begin{align}
        N\subset B_e(3d_k/4+b_k)\subset B_e(d_{k+1}/4),
    \end{align}
    where the factor $3d_k/4+b_k$ arises because $e$ and $e'$ are $(d_k/2, b_k+d_k/2)$-linked.

    We then show Eq.~\eqref{eq:cardinality_of_N} similarly to Eq.~\eqref{eq:radius_bound_N} by induction on $k$.
    Since $M_0(e) = \{\{e\}\}$, Eq.~\eqref{eq:cardinality_of_N} holds for $k=0$.
    Assuming that Eq.~\eqref{eq:cardinality_of_N} holds for $k$, we will show Eq.~\eqref{eq:cardinality_of_N} for $k+1$.
    In the decomposition~\eqref{eq:decomposition_N}, since $N_1\subset B_e(d_k/4)$, $N_2\subset B_{e'}(d_k/4)$, and $\dist(e,e')> d_k/2$ hold, we have $N_1\cap N_2 = \varnothing$.
    Therefore, we obtain
    \begin{align}
        \abs{N} = \abs{N_1}+\abs{N_2}.
    \end{align}
    By assumption, $\abs{N_1}=\abs{N_2} = 2^k$ holds.
    Thus, we obtain
    \begin{align}
        \abs{N} = 2^{k+1}.
    \end{align}

    We finally show Eq.~\eqref{eq:cardinality_Mk}.
    Since any $N\in M_k(e)$ is given by $N = N_1\cup N_2$ such that $N_1\in M_{k-1}(e)$ and $N_2\in M_{k-1}(e')$, and $e$ and $e'$ are $(d_{k-1}/2, b_{k-1}+d_{k-1}/2)$-linked, we obtain
    \begin{align}
        \abs{M_k(e)}
        &= \abs{M_{k-1}(e)} \sum_{e'\in B_e(b_{k-1}+d_{k-1}/2)\setminus B_e(d_{k-1}/2)} \abs{M_{k-1}(e')}\\
        &\leq \abs{B_e(b_{k-1}+d_{k-1}/2)} \left(\max_{e\in E} \abs{M_{k-1}(e)}\right)^2\\
        &\leq \Lambda (b_{k-1}+d_{k-1}/2)^\Delta \left(\max_{e\in E} \abs{M_{k-1}(e)}\right)^2.
    \end{align}
    Since $\abs{M_0(e)}=1$ holds, we obtain
    \begin{align}
        \abs{M_k(e)} &\leq \prod_{j=0}^{k-1} \left[\Lambda \qty(b_{k-j-1}+\frac{d_{k-j-1}}{2})^\Delta\right]^{2^{j}}.
    \end{align}
\end{proof}

\end{document}